\documentclass[final]{IEEEtran}
\usepackage{amssymb}
\usepackage{amsfonts}
\usepackage{amsmath}
\usepackage{amsthm}
\usepackage{tikz}
\usepackage{hyperref}

\newcommand{\func}{\operatorname}
\theoremstyle{plain}
\newtheorem{theorem}{Theorem}

\newtheorem{conjecture}[theorem]{Conjecture}
\newtheorem{corollary}[theorem]{Corollary}
\newtheorem{lemma}[theorem]{Lemma}

\newtheorem{proposition}[theorem]{Proposition}
\theoremstyle{definition}
\newtheorem{definition}[theorem]{Definition}
\newtheorem{example}[theorem]{Example}
\theoremstyle{remark}
\newtheorem{remark}[theorem]{Remark}

\providecommand{\MR}{\relax\ifhmode\unskip\space\fi MR }

\providecommand{\href}[2]{#2}

\begin{document}

\title{50 Years of the Golomb--Welch Conjecture}

\author{Peter~Horak and Dongryul~Kim
  \thanks{Peter~Horak is with SIAS, University of Washington, Tacoma, WA 98402 USA (email: \href{mailto:horak@uw.edu}{horak@uw.edu}).}
  \thanks{Dongryul~Kim is with Department of Mathematics, Harvard University, Cambridge, MA 02138 USA (email: \href{mailto:dkim04@college.harvard.edu}{dkim04@college.harvard.edu}).}
  \thanks{Copyright \copyright{} 2017 IEEE. Personal use of this material is permitted. However, permission to use this material for any other purposes must be obtained from the IEEE by sending a request to \href{mailto:pubs-permission@ieee.org}{pubs-permission@ieee.org}}
}

\maketitle

\begin{abstract}
Since 1968, when the Golomb--Welch conjecture was raised, it has become the
main motive power behind the progress in the area of the perfect Lee codes.
Although there is a vast literature on the topic and it is widely believed
to be true, this conjecture is far from being solved. In this paper, we
provide a survey of papers on the Golomb--Welch conjecture. Further, new
results on Golomb--Welch conjecture dealing with perfect Lee codes of large
radii are presented. Algebraic ways of tackling the conjecture in the future
are discussed as well. Finally, a brief survey of research inspired by the
conjecture is given.
\end{abstract}


\section{Introduction}

\label{sec:introduction}

In this paper we deal with codes in the Lee metric. This metric was
introduced in \cite{Lee} and \cite{Ulrich} for transmission of signals taken
from $GF(p)$ over noisy channels. It was generalized for $\mathbb{Z}_{m}$ in 
\cite{GW}. The interest in Lee codes is due to many applications of them.
For example, constrained and partial-response channels \cite{Roth}, flash
memory \cite{Sch}, interleaving schemes \cite{Blaum}, placement of resources
in the computer architecture that minimizes access time by processing
elements \cite{B}, multidimensional burst-error-correction \cite{Etzion},
and error-correction in the rank modulation scheme for flash memories \cite%
{Jiang}.

50 years ago, Golomb and Welch \cite{GW} raised a conjecture on the
existence of perfect $e$-error-correcting codes in the Lee metric. This
conjecture lies at the very center of interests in the area of perfect codes
in the Lee metric. In spite of great effort and plenty of papers on the
topic, the Golomb--Welch conjecture is still far from being solved. In
Section~\ref{sec:GW} we survey results on this conjecture. As a main part of
the paper, new results on the conjecture are provided in Sections~\ref%
{sec:PostLep} and \ref{sec:numbertheoretic}. Namely, Golomb and Welch proved
that for each fixed $n$ there exists an $e_{n}$, $e_{n}$ unspecified, such
that for all $e>e_{n}$ there is no perfect $e$-error correcting code in $%
\mathbb{Z}^{n}$ with the Lee metric. In Section 3 we present the first
explicit upper bound on $e_{n}$. More precisely, we show that the condition 
\emph{periodic} can be dropped in the Post \cite{P} and Lepist\"{o} \cite{Le}
bounds (see Theorem \ref{thm:Postextend} and Theorem \ref{thm:Lepextend}).
Finally, we exhibit how a linear programming technique can be used to obtain
another bound on $e_{n},$ cf. Corollary \ref{cor:LP}. Combining these three
statements we obtained Theorem \ref{thm:GWsummary} that summarizes our new
results on the Golomb--Welch conjecture.

Although the conjecture has been tackled in various ways, using different
techniques, it seems to us that none of them is powerful enough to entirely
solve the conjecture. We believe that a new approach has to be developed.
Therefore, possible avenues how to attack the conjecture are discussed in
Section 4. Using the so-called polynomial method, a necessary condition for 
the existence of a tiling of $\mathbb{Z}^{n}$ by translates of a tile $V$ is 
proved (see Theorem~\ref{D}). We guess that this is a first necessary
condition for a generic (arbitrary) tile. Further, we exhibit usage of
Fourier analysis in this area; we provide a sufficient condition for a tile $%
V$ such that each translational tiling of $\mathbb{Z}^{n}$ by $V$ is
periodic (see Theorem \ref{thm:Fourier}). In our quest to prove the
Golomb--Welch conjecture we have dealt with tiles of prime size. Later we
started to be interested in these tiles on it own right. Now it seems that a
part of our research on prime tiles might contribute back to the
Golomb--Welch conjecture. In this regard, first we reprove a statement that
each tiling of $\mathbb{Z}^{n}$ by translates of a tile of prime size has to
be periodic. In fact, we conjecture that each such tiling has to be even a
lattice one (see Conjecture \ref{conj:primetile}). We prove our conjecture
for tiles of size at most $7$.

In Section~\ref{sec:further} we cover results inspired by the Golomb--Welch
conjecture. First we describe several generalizations and modification of
the conjecture, and then a brief survey of the results on quasi-perfect Lee
codes will be given.

In the last section we summarize our discussion on the methods used and the
methods proposed in this paper to solve the Golomb--Welch conjecture.

\subsection{Terminology and Basic Concepts}

As usual, let $\mathbb{Z}$ be the set of all integers, $\mathbb{Z}_{q}$
denote the ring of integers modulo $q$, and let $T^{n}$ stand for the $n$%
-fold Cartesian product of a set $T$. A Lee code is a subset of the metric
space $(\mathcal{C},\delta _{L})$, where $\mathcal{C}=\mathbb{Z}_{q}^{n}$,
or $\mathcal{C}=\mathbb{Z}^{n}$, and $\delta _{L}$ is the Lee metric (= the
Manhattan metric, the zig-zag metric, the $\ell ^{1}$-norm). That is, for
any two words $\mathbf{u}=(u_{1},u_{2},\ldots ,u_{n})$ and $\mathbf{v}%
=(v_{1},v_{2},\ldots ,v_{n})$, 
\begin{align*}
\delta _{L}(\mathbf{u},\mathbf{v})& =\sum_{i=1}^{n}\min (\lvert
u_{i}-v_{i}\rvert ,q-\lvert u_{i}-v_{i}\rvert )\text{ for }\mathbf{u},%
\mathbf{v}\in \mathbb{Z}_{q}^{n}, \\
\delta _{L}(\mathbf{u},\mathbf{v})& =\sum_{i=1}^{n}\lvert u_{i}-v_{i}\rvert 
\text{ for }\mathbf{u},\mathbf{v}\in \mathbb{Z}^{n}.
\end{align*}%
A Lee code $C$ is an \emph{$e$-error-correcting code} if any two distinct
elements of $C$ have distance at least $2e+1$. An $e$-error-correcting Lee
code is further called \emph{perfect} if for each $\mathbf{x}\in \mathbb{Z}%
_{q}^{n}$ ($\mathbf{x}\in \mathbb{Z}^{n}$), there exists a unique element $%
\mathbf{c}\in C$ such that $\delta _{L}(\mathbf{x},\mathbf{c})\leq e$. A
perfect $e$-error-correcting Lee code in $\mathbb{Z}_{q}^{n}$ and in $%
\mathbb{Z}^{n}$ will be called $PL(n,e,q)$ and $PL(n,e)$, respectively.
These codes are also termed perfect $e$-error-correcting code of block size $%
n$ over $\mathbb{Z}_{q}$ (over $\mathbb{Z}$). If $q\geq 2e+1$, a $PL(n,e,q)$%
-code is said to be over a \emph{large} alphabet, otherwise it is said to be
over a \emph{small} alphabet. \ A set $S\subset \mathbb{Z}^{n}$ is $q$%
-periodic if it is periodic with the period $q$ along all coordinate axes. A 
$PL(n,e)$-code $C$ is ($q$-) \emph{periodic} (resp.\ \emph{lattice}, \emph{%
linear}) if $C$ is a ($q$-) periodic set in $\mathbb{Z}^{n}$ (resp.\ a
subgroup of the additive group $\mathbb{Z}^{n}$ of full rank).

It is very common to define error-correcting Lee codes using the language of
tilings. In this setting it is not difficult to see that to know all about $%
PL(n,e,q)$-codes with large alphabets, it suffices to study $PL(n,e)$-codes.
Indeed, consider the \emph{Lee spheres} 
\begin{align*}
S(n,e,q)& =\{\mathbf{x}\in \mathbb{Z}_{q}^{n}:\delta _{L}(\mathbf{x},\mathbf{%
0})\leq e\}\text{ and } \\
S(n,e)& =\{\mathbf{x}\in \mathbb{Z}^{n}:\delta _{L}(\mathbf{x},\mathbf{0}%
)=\lvert x_{1}\rvert +\cdots +\lvert x_{n}\rvert \leq e\}
\end{align*}%
of radius $e$. Then $PL(n,e,q)$-codes and periodic $PL(n,e)$-codes can be
naturally identified with tilings of $\mathbb{Z}_{q}^{n}$ and of $\mathbb{Z}%
^{n}$ by translates of $S(n,e,q)$ and $S(n,e)$, respectively.

If $q\geq 2e+1$, then the natural projection map $\mathbb{Z}^{n}\rightarrow 
\mathbb{Z}_{q}^{n}$ restricts to a bijection from $S(n,e)$ to $S(n,e,q)$.
Any tiling of $\mathbb{Z}_{q}^{n}$ by $S(n,e,q)$ will then pull back via the
projection to a periodic tiling of $\mathbb{Z}^{n}$ by $S(n,e)$. Then a $%
PL(n,e,q)$-code induces a periodic $PL(n,e)$-code that is a disjoint union
of cosets of $q\mathbb{Z}^{n}\subset \mathbb{Z}^{n}$. Conversely, any such
periodic $PL(n,e)$-code clearly comes from a $PL(n,e,q)$-code. The following
proposition states in a formal way that $PL(n,e)$-codes provide full
information about $PL(n,e,q)$-codes.

\begin{proposition}
\label{P1} For $q\geq 2e+1$, there exists a natural bijection between $%
PL(n,e,q)$-codes and $q$-periodic $PL(n,e)$-codes that is a union of cosets
of $q\mathbb{Z}^{n}\subset \mathbb{Z}^{n}$, given by taking the image or the
inverse image with respect to the projection map $\mathbb{Z}^{n}\rightarrow 
\mathbb{Z}_{q}^{n}$.
\end{proposition}

\noindent Since a $PL(n,e)$-code can be seen as a partition of $\mathbb{Z}%
^{n}$, only a small step is needed to get a geometrical interpretation of $%
PL(n,e)$-codes. Let $\mathbb{R}$ be the set of real numbers. Consider the $n$%
-dimensional space $\mathbb{R}^{n}$ endowed with the Lee metric $\delta _{L}$%
. The \emph{$n$-cube} centered at $\mathbf{x}=(x_{1},\ldots ,x_{n})\in 
\mathbb{R}^{n}$ is the set $C(\mathbf{x})=\{\mathbf{y}=(y_{1},\ldots
,y_{n}):\left\vert y_{i}-x_{i}\right\vert \leq \frac{1}{2}\}$. By a Lee
sphere of radius $e$ in $\mathbb{R}^{n}$ centered at $\mathbf{0}$, $L(n,e)$,
we understand the union of $n$-cubes centered at $y$, where $\delta _{L}(%
\mathbf{y},\mathbf{0})\leq e$, and $\mathbf{y}$ has integer coordinates.
Finally, a Lee sphere of radius $e$ in $\mathbb{R}^{n}$ centered at $\mathbf{%
x}\in \mathbb{R}^{n}$ is the set $x+L(n,e)=\{x+l:l\in L(n,e)$. Clearly, a $%
PL(n,e)$-code exists if and only if there is a tiling of $\mathbb{R}^{n}$ by
Lee spheres $L(n,e)$. The Lee spheres $L(2,1),L(2,2),L(3,1)$, and $L(3,2)$
are depicted in Figure~\ref{fig:1}. The advantage of understanding a $%
PL(n,e) $-code as a tiling of $\mathbb{R}^{n}$ of $L(n,e)$ is in the
possibility of applying deep results in geometry to Lee codes.

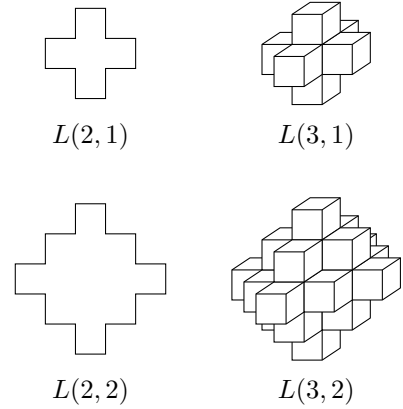
\begin{figure}[th]
\centering
\begin{tikzpicture}
    \begin{scope}[scale=0.4]
      \draw (-0.5,-0.5) -- (-0.5,-1.5) -- (0.5,-1.5) -- (0.5,-0.5) -- (1.5,-0.5) -- (1.5,0.5) -- (0.5,0.5) -- (0.5,1.5) -- (-0.5,1.5) -- (-0.5,0.5) -- (-1.5,0.5) -- (-1.5,-0.5) -- cycle;
    \end{scope}
    \node at (0,-1.1) {$L(2,1)$};

    \begin{scope}[shift={(0,-3)},scale=0.4]
      \draw (-1.5,-0.5) -- (-1.5,-1.5) -- (-0.5,-1.5) -- (-0.5,-2.5) -- (0.5,-2.5) -- (0.5,-1.5) -- (1.5,-1.5) -- (1.5,-0.5) -- (2.5,-0.5) -- (2.5,0.5) -- (1.5,0.5) -- (1.5,1.5) -- (0.5,1.5) -- (0.5,2.5) -- (-0.5,2.5) -- (-0.5,1.5) -- (-1.5,1.5) -- (-1.5,0.5) -- (-2.5,0.5) -- (-2.5,-0.5) -- cycle;
    \end{scope}
    \node at (0,-4.5) {$L(2,2)$};

    \begin{scope}[shift={(3,0)},scale=0.4]
      \foreach \x/\y/\z in {-1/0/0, 0/1/0, 0/0/-1, 1/0/0, 0/-1/0, 0/0/1} {
	\draw[fill=white] (\x+0.6*\y-0.8, 0.4*\y+\z+0.3) -- (\x+0.6*\y-0.2, 0.4*\y+\z+0.7) -- (\x+0.6*\y+0.8, 0.4*\y+\z+0.7) -- (\x+0.6*\y+0.8, 0.4*\y+\z-0.3) -- (\x+0.6*\y+0.2, 0.4*\y+\z-0.7) -- (\x+0.6*\y-0.8, 0.4*\y+\z-0.7) -- cycle;
	\draw (\x+0.6*\y-0.8, 0.4*\y+\z+0.3) -- (\x+0.6*\y+0.2, 0.4*\y+\z+0.3) -- (\x+0.6*\y+0.8, 0.4*\y+\z+0.7);
	\draw (\x+0.6*\y+0.2, 0.4*\y+\z+0.3) -- (\x+0.6*\y+0.2, 0.4*\y+\z-0.7);
      }
    \end{scope}
    \node at (3,-1.1) {$L(3,1)$};
    
    \begin{scope}[shift={(3,-3)},scale=0.4]
      \foreach \x/\y/\z in {-2/0/0, 0/2/0, 0/0/-2, -1/1/0, 0/1/-1, -1/0/-1, 1/1/0, -1/-1/0, 0/1/1, 0/-1/-1, 1/0/-1, -1/0/1, 2/0/0, 0/-2/0, 0/0/2, 1/-1/0, 1/0/1, 0/-1/1} {
	\draw[fill=white] (\x+0.6*\y-0.8, 0.4*\y+\z+0.3) -- (\x+0.6*\y-0.2, 0.4*\y+\z+0.7) -- (\x+0.6*\y+0.8, 0.4*\y+\z+0.7) -- (\x+0.6*\y+0.8, 0.4*\y+\z-0.3) -- (\x+0.6*\y+0.2, 0.4*\y+\z-0.7) -- (\x+0.6*\y-0.8, 0.4*\y+\z-0.7) -- cycle;
	\draw (\x+0.6*\y-0.8, 0.4*\y+\z+0.3) -- (\x+0.6*\y+0.2, 0.4*\y+\z+0.3) -- (\x+0.6*\y+0.8, 0.4*\y+\z+0.7);
	\draw (\x+0.6*\y+0.2, 0.4*\y+\z+0.3) -- (\x+0.6*\y+0.2, 0.4*\y+\z-0.7);
      }
    \end{scope}
    \node at (3,-4.5) {$L(3,2)$};
  \end{tikzpicture}
\caption{Figure of $L(2,1)$, $L(2,2)$, $L(3,1)$, and $L(3,2)$}
\label{fig:1}
\end{figure}

We note that it is common in Coding theory to call the sets $S(n, e) = \{ 
\mathbf{x} \in \mathbb{Z}^n : d(\mathbf{x}, \mathbf{0}) \leq e\}$ and $B(n,
e) = \{\mathbf{x} \in \mathbb{Z}^n : d(\mathbf{x}, \mathbf{0}) = e \}$ the
sphere and the boundary of this sphere although in other parts of
mathematics they are termed the ball and the sphere. In order not to go
against the long time tradition we have decided to stick with this imprecise
terminology.

\section{The Golomb--Welch Conjecture}

\label{sec:GW}

In this section we state the Golomb--Welch conjecture and survey related
results. We do not cover here $PL(n,e,q)$-codes over small alphabets.

In their seminal paper Golomb and Welch \cite{GW} discuss at great length
the existence of $PL(n,e,q)$-codes. They constructed $PL(n,e,q)$-codes for
parameters $(n,e,q)=(1,e,2e+1)$, $(2,e,e^{2}+(e+1)^{2})$, and $(n,1,2n+1)$.
In the last paragraph of Section 3 in \cite{GW} it is conjectured that there
are no tilings of $\mathbb{Z}_{q}^{n}$ by Lee spheres over large alphabet
for other values of $(n,e).$ We note that in \cite{GW} $\mathbb{Z}_{q}^{n}$
is called $n$-dimensional space while $\mathbb{Z}^{n}$ is termed $n$%
-dimensional Euclidean space.

\begin{conjecture}[Golomb--Welch, weak version, Section~3 \protect\cite{GW}]

\label{conj:weakGW} There is no $PL(n,e,q)$-code over large alphabets for $%
n\geq 3$ and $e\geq 2$.
\end{conjecture}

In Section 7, Golomb and Welch formulate their conjecture in terms of tiling 
$n$-dimensional Euclidean space. Thus, with respect to Proposition \ref{P1},
the following conjecture is a natural strengthening:

\begin{conjecture}[Golomb--Welch, strong version, Section~7 \protect\cite{GW}%
]
\label{conj:strongGW} There is no $PL(n,e)$-code for $n\geq 3$ and $e\geq 2$.
\end{conjecture}

A set $S\subset \mathbb{Z}^{n}$ is fully periodic if the set of those
elements which shift $S$ into itself is a subgroup of $\mathbb{Z}^{n}$ of
finite index. We point out that if the following Lagarias--Wang conjecture
is true then Conjectures~\ref{conj:weakGW} and \ref{conj:strongGW} are
equivalent.

\begin{conjecture}[Lagarias--Wang \protect\cite{Lagarias}]
\label{conj:periodictiling} If $V$ tiles $\mathbb{Z}^{n}$ by translations,
then $V$ admits a fully periodic tiling, i.e., a $q$-periodic tiling for
sufficiently large $q$.
\end{conjecture}

\noindent So far Conjecture~\ref{conj:periodictiling} has been proved for
tiles $V$ of prime size \cite{Szegedy}, for any $V\subset \mathbb{Z}^{2}$ 
\cite{Ba}, and for some other special types of tiles.

\subsection{Survey of Results on the Golomb--Welch Conjecture}

To provide a support for their conjecture, Golomb and Welch \cite{GW} show
that there is no $PL(n,e)$-code for $(n,e) = (3,2)$ and also for large $e$.
Their basic idea for proving nonexistence of $PL(n,e)$-codes for
sufficiently large $e$ is that such a code will induce a dense packing of $%
\mathbb{R}^n$ by cross-polytopes. The following theorem then follows from
the known fact that there is no tiling of $\mathbb{R}^{n}$ by regular
cross-polytopes.

\begin{theorem}[\protect\cite{GW}]
\label{thm:tempGW} For $n \ge 3$ there exists $e_n$, $e_n$ not specified,
such that for any $e > e_n$ there is no $PL(n,e)$-code.
\end{theorem}

For a more detailed explanation of their idea, see the beginning of Section~%
\ref{sec:PostLep}. Theorem~\ref{thm:tempGW} is not explicit, and not even
effective in the sense that it only shows that such a constant $e_{n}$
exists. A first explicit bound on $e_{n}$, in the case of periodic codes,
has been given by Post \cite{P}. He showed, by counting low-dimensional
cross-sections, that $PL(n,e,q)$-codes do not exist for $3\leq n\leq 5$, $%
e\geq n-1$, $q\geq 2e+1$ and $n\geq 6$, $e\geq \frac{\sqrt{2}}{2}n-\frac{3}{4%
}\sqrt{2}-\frac{1}{2}$, $q\geq 2e+1$. The result of Post was asymptotically
improved by Astola \cite{A2}, and later by Lepist\"{o} \cite{Le} who
obtained:

\begin{theorem}[\protect\cite{Le}]
\label{thm:Lep} For any $n,e,q$ satisfying $n<(e+2)^{2}/2.1,e\geq 285$, and $%
q\geq 2e+1$, there is no $PL(n,e,q)$-code.
\end{theorem}

An outline of Post's and Lepisto's proofs will be provided in the next
section. Developing and refining their ideas we will show that the condition 
\emph{periodic} can be dropped from both their results, cf. Theorem \ref%
{thm:Postextend}, and Theorem \ref{thm:Lepextend}. Also, by using a linear
programming method, we obtained a further slight improvement on the bound of 
$e_{n}$( see Corollary \ref{cor:LP}). The next theorem is a direct
combination of these three bounds on $e_{n}.$

\begin{theorem}
\label{thm:GWsummary}There is no $PL(n,e)$-code for%
\begin{align*}
3 &\leq n\leq 74 &\text{ and } & \max\Bigl\{ \frac{\sqrt{2}}{2}n - 
\frac{3}{4}\sqrt{2}-\frac{1}{2}, 2 \Bigr\} \le e, \\
75 &\leq n\leq 405 &\text{ and } & \max\{18, \sqrt{2n+40}\} \leq e\leq \frac{n-21}{3}\\
& & &\text{ or }\frac{\sqrt{2}}{2}n-\frac{3}{4}\sqrt{2}-\frac{1}{2}\leq e, \\
406 &\leq n\leq 876 &\text{ and } & \sqrt{2n+40}\leq e\leq \frac{n-21}{3}\text{
or }285\leq e, \\
876 &\leq n &\text{ and } & \sqrt{2n+40} \le e.
\end{align*}
\end{theorem}

It seems that the most difficult case of the Golomb--Welch conjecture is
that of $e=2$. The nonexistence of $PL(6,2)$-codes has been shown in \cite%
{H5}. A step forward in this direction has been made by the second author of
this paper (see \cite{Kim}). He proved that if the volume of the sphere $%
\lvert S(n,2)\rvert =2n^{2}+2n+1$ is prime and a certain number-theoretic
condition is satisfied, then $PL(n,2)$-codes do not exist. It turns out that
this condition is not restrictive as, e.g., out of $12706$ numbers $n\leq
10^{5}$ with $p=2n^{2}+2n+1$ prime, only $4$ numbers $n$ do not satisfy the
condition. However, it is not known if there are infinity many $n$ with $%
p=2n^{2}+2n+1$ prime.

A special case, the nonexistence of linear $PL(n,2)$-codes is proved in \cite%
{H6} for $n\leq 12$. The proof is based on the nonexistence of a
homomorphism $\phi :\mathbb{Z}^{n}\rightarrow G$, an abelian group of order $%
\lvert S(n,2)\rvert $ such that a restriction of $\phi $ to $S(n,2)$ would
be a bijection to $G$. A similar approach has been used in \cite{Zhang} to
show the nonexistence of linear $PL(n,3)$-codes for some values of $n\equiv
12,21(\func{mod}27)$, and the nonexistence of linear $PL(n,4)$-codes for
some values of $n\equiv 3,5,21,23(\func{mod}27)$.

As to the weak version of the Golomb--Welch conjecture, Conjecture~\ref%
{conj:weakGW}, the nonexistence of $PL(n,e,q)$-codes has been proved for
several special cases of $q$. A list of such cases is given in \cite{A1}. To
illustrate this type of conditions, here we mention two of them (see \cite%
{A1}): There is no $PL(n,e,q)$-code for $e=2,q=p^{k}$, $p$ is a prime, $p\neq
13$, $p<\sqrt{\lvert S(n,2)\rvert }$; and $e=3$, $q\geq 7$ is not divisible
by a prime $p\equiv 1,3,5,7,9(\func{mod}20)$. 

Now we turn our attention to the case of small dimension $n$. For $3\leq
n\leq 5$, the Golomb--Welch conjecture has been proved for all $e\geq 2$. In 
\cite{GMP}, by an elegant \textquotedblleft{}picture says it all%
\textquotedblright{} approach it is shown that there is no tiling of $%
\mathbb{R}^{3}$ by Lee spheres. A further extension of the result has been
provided in \cite{GMP1} (see Section~\ref{sec:further}). The same result,
using an exhaustive computer search, was proved in \cite{Spa} for $\mathbb{R}%
^{4}$. It seems that the used algorithm is not computationally feasible for $%
n\geq 5 $. Finally, by an algebraic approach based on the nonexistence of $%
PL(n,2)$-codes, it was proved analytically that, for $3\leq n\leq 5$, there
is no tiling of $\mathbb{R}^{n}$ by Lee spheres \cite{H2}.

\section{The Golomb--Welch Conjecture for Large Radius}

\label{sec:PostLep}

In this section we study ways of proving the nonexistence of $PL(n,e)$%
-codes, in the case when $e$ is sufficiently large. Why would $e$ being
large prevent the Lee sphere $S(n,e)$ from tiling $\mathbb{Z}^{n}$? The
intuition is that as $e$ grows for fixed $n$, the sphere $S(n,e)$ becomes
more and more similar to the convex hull of $\{(0,\ldots ,0,\pm 1,0,\ldots
,0)\}$. This polytope is the dual of the $n$-cube, and is called a \emph{%
cross-polytope}.

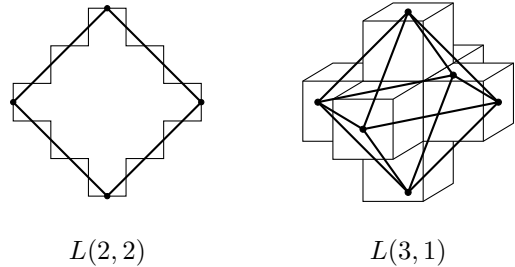
\begin{figure}[th]
\centering
\begin{tikzpicture}
    \begin{scope}[shift={(0,0)},scale=0.5]
      \draw (-1.5,-0.5) -- (-1.5,-1.5) -- (-0.5,-1.5) -- (-0.5,-2.5) -- (0.5,-2.5) -- (0.5,-1.5) -- (1.5,-1.5) -- (1.5,-0.5) -- (2.5,-0.5) -- (2.5,0.5) -- (1.5,0.5) -- (1.5,1.5) -- (0.5,1.5) -- (0.5,2.5) -- (-0.5,2.5) -- (-0.5,1.5) -- (-1.5,1.5) -- (-1.5,0.5) -- (-2.5,0.5) -- (-2.5,-0.5) -- cycle;
      \draw[fill] (-2.5,0) circle[radius=0.07];
      \draw[fill] (2.5,0) circle[radius=0.07];
      \draw[fill] (0,-2.5) circle[radius=0.07];
      \draw[fill] (0,2.5) circle[radius=0.07];
      \draw[thick] (-2.5,0) -- (0,-2.5) -- (2.5,0) -- (0,2.5) -- cycle;
    \end{scope}
    \node at (0,-2) {$L(2,2)$};

    \begin{scope}[shift={(4,0)},scale=0.8]
      \foreach \x/\y/\z in {-1/0/0, 0/1/0, 0/0/-1, 1/0/0, 0/-1/0, 0/0/1} {
	\draw[fill=white] (\x+0.5*\y-0.75, 0.3*\y+\z+0.35) -- (\x+0.5*\y-0.25, 0.3*\y+\z+0.65) -- (\x+0.5*\y+0.75, 0.3*\y+\z+0.65) -- (\x+0.5*\y+0.75, 0.3*\y+\z-0.35) -- (\x+0.5*\y+0.25, 0.3*\y+\z-0.65) -- (\x+0.5*\y-0.75, 0.3*\y+\z-0.65) -- cycle;
	\draw (\x+0.5*\y-0.75, 0.3*\y+\z+0.35) -- (\x+0.5*\y+0.25, 0.3*\y+\z+0.35) -- (\x+0.5*\y+0.75, 0.3*\y+\z+0.65);
	\draw (\x+0.5*\y+0.25, 0.3*\y+\z+0.35) -- (\x+0.5*\y+0.25, 0.3*\y+\z-0.65);
      }
      \draw[fill] (1.5,0) circle[radius=0.05];
      \draw[fill] (-1.5,0) circle[radius=0.05];
      \draw[fill] (0,1.5) circle[radius=0.05];
      \draw[fill] (0,-1.5) circle[radius=0.05];
      \draw[fill] (0.75,0.45) circle[radius=0.05];
      \draw[fill] (-0.75,-0.45) circle[radius=0.05];
      \draw[thick] (1.5,0) -- (0,-1.5) -- (-1.5,0) -- (0,1.5) -- cycle;
      \draw[thick] (1.5,0) -- (0.75,0.45) -- (-1.5,0) -- (-0.75,-0.45) -- cycle;
      \draw[thick] (0.75,0.45) -- (0,-1.5) -- (-0.75,-0.45) -- (0,1.5) -- cycle;
    \end{scope}
    \node at (4,-2) {$L(3,1)$};
  \end{tikzpicture}
\caption{Figure of a cross-polytope in $\mathbb{R}^{2}$ and $\mathbb{R}^{3}$}
\label{fig:2}
\end{figure}

For $n\geq 3$, it is well-known that the cross-polytope does not tile $%
\mathbb{R}^{n}$ by translations. If $n\neq 4$, then this fact can be
immediately obtained by computing the angle of two adjacent faces, and for $%
n=4$, taking care of the orientation gives this fact. Using a compactness
argument on the space of local configurations, it can be shown that the
packing density of a bounded set that does not tile $\mathbb{R}^{n}$ is
bounded away from $1$. On the other hand, a $PL(n,e)$-code (even a $QPL(n,e)$%
-code, see Section~\ref{subsec:QPL} for the definition of $QPL(n,e)$-code)
for large $e$ induces a translational packing of a $n$-dimensional
cross-polytope with high density. Thus we obtain Theorem~\ref{thm:tempGW}.

\newcounter{tempthm} \setcounter{tempthm}{\arabic{theorem}} %
\setcounter{theorem}{4}

\begin{theorem}[\protect\cite{GW}]
For $n \ge 3$ there exists $e_n$, $e_n$ not specified, such that for any $e
> e_n$ there is no $PL(n,e)$-code.
\end{theorem}

\setcounter{theorem}{\arabic{tempthm}}

\begin{remark}
We note that in \cite{GW}, the theorem was actually proved only for $n=3$
and $n\geq 5$. However, it is not difficult to see that the same argument
holds also for $n=4.$ Indeed, although there is a tiling of $\mathbb{R}^{4}$
by the $4$-dimensional cross-polytope, known as the $16$-cell honeycomb,
this tiling is not by translations.
\end{remark}

The idea of the proof of Theorem~\ref{thm:tempGW} has been used by several
authors (see e.g.~\cite{H6}), where, applying the idea, it is proved that for
each $n$ there are only finitely many values of $e$ for which quasi-perfect
Lee code might exist.

\subsection{Post's Bound}

The first ever effective result on the nonexistence of $PL(n,e,q)$-codes for
large $e$ was obtained by Post.

\begin{theorem}[\protect\cite{P}]
\label{thm:Post} For any $n,e,q$ satisfying $n\geq 6$, $e\geq \frac{\sqrt{2}%
}{2}n-\frac{3}{4}\sqrt{2}-\frac{1}{2}$, and $q\geq 2e+1$, there is no $%
PL(n,e,q)$-code.
\end{theorem}

Post obtained this theorem by focusing on local configurations of the tiling
at the boundary of the Lee spheres. Let us be more specific.

\begin{definition}
A \emph{$k$-dimensional sector} is a subset of $\mathbb{Z}^n$ or $\mathbb{Z}%
_q^n$ that is a translate of 
\begin{equation*}
\{ \mathbf{x} : x_i \in \{0, 1\} \text{ if } i \in \{ i_1, \ldots, i_k \},
\; x_i = 0 \text{ otherwise} \},
\end{equation*}
where $\{i_1, \ldots, i_k\} \subseteq \{1, \ldots, n\}$.
\end{definition}

Clearly, a $k$-dimensional sector has cardinality $2^{k}$, and a non-empty
intersection of a $k$-dimensional sector and a Lee sphere always has
cardinality $\sum_{i=0}^{t}\binom{k}{i}$ for some $0\leq t\leq k$. In \cite%
{P}, Post focused on the $6$-dimensional sectors in $\mathbb{Z}_{q}^{n}$. A $%
6$-dimensional sector consists of $64$ unit cubes (or points if we agree to
work in $\mathbb{Z}^{n}$), and a Lee sphere can cover either $%
0,1,7,22,42,57,63,64$ of them. Ignore the case when the $6$-dimensional
sector is disjoint from or entirely covered by the Lee sphere. Let us say
that a pair $(S,T)$ of a Lee sphere $S$ and a $6$-dimensional sector $T$ is
of \emph{type $i$} if $\lvert S\cap T\rvert =i$. For a fixed sector $T$,
there are only a handful of ways it can be covered completely: one type~$63$
and one type~$1$, one type~$42$ and one type~$7$ and fifteen type~$1$, etc.
Given a tiling of $\mathbb{Z}^{n}$ by Lee spheres $S(n,e)$. We denote by $%
t_{T,i}$ the number of type~$i$ pairs $(S,T)$ with given sector $T$. Listing
all possible combinations, Post proves the following.

\begin{lemma}[{\protect\cite[p.~377]{P}}]
\label{lem:Post1} Given a tiling of $\mathbb{Z}^{n}$ by $S(n,e)$, $n\geq 6$,
for any $6$-dimensional sector $T$, 
\begin{equation*}
t_{T,1}-t_{T,7}-10t_{T,22}+10t_{T,42}+t_{T,57}-t_{T,63}\geq 0.
\end{equation*}
\end{lemma}

On the other hand, let us count the number of type~$i$ pairs for a fixed Lee
sphere $S$. We know quite well the shape of $S(n,e)$, and it is a matter of
computation to count $6$-dimensional sectors $T$ for which $\lvert
S(n,e)\cap T\rvert =i$. Denote by $g_{i}$ the number of type~$i$ pairs $%
(S(n,e),T)$.\footnote{%
In \cite{P}, $g_{i}$ denotes the number of $6$-dimensional sectors in one
orthant. As a consequence, the $g_{i}$ used here differs from Post's $g_{i}$
by a factor of $64$. However this does not affect anything.}

\begin{lemma}[{\protect\cite[p.~378--379]{P}}]
\label{lem:Post2} If $n \ge 6$ and $e \ge \frac{\sqrt{2}}{2} n - \frac{3}{4} 
\sqrt{2} - \frac{1}{2}$, then 
\begin{equation*}
g_1 - g_7 - 10 g_{22} + 10 g_{42} + g_{57} - g_{63} < 0.
\end{equation*}
\end{lemma}

Suppose now that $C$ is a $PL(n,e,q)$-code, for $q\geq 2e+1$ and $n\geq 6$, $%
e\geq \frac{\sqrt{2}}{2}n-\frac{3}{4}\sqrt{2}-\frac{1}{2}$. Let $t_{i}$ be
the total number of type~$i$ pairs $(S,T)$ in $\mathbb{Z}_{q}^{n}$, which is
clearly finite. Then we count $t_{i}$ in two ways as 
\begin{equation*}
t_{i}=\sum_{6\text{-dim.\ sect.\ }T}^{{}}t_{T,i}=\lvert C\rvert \cdot g_{i}.
\end{equation*}%
Thus from Lemmas~\ref{lem:Post1} and \ref{lem:Post2} it follows that 
\begin{align*}
 t_{1}&-t_{7}-10t_{22}+10t_{42}+t_{57}-t_{63} \\
&
=\sum_{T}^{{}}(t_{T,1}-t_{T,7}-10t_{T,22}+10t_{T,42}+t_{T,57}-t_{T,63})\geq
0, \\
t_{1}&-t_{7}-10t_{22}+10t_{42}+t_{57}-t_{63} \\
& =\lvert C\rvert (g_{1}-g_{7}-10g_{22}+10g_{42}+g_{57}-g_{63})<0,
\end{align*}%
This is clearly a contradiction, and hence Theorem~\ref{thm:Post} is proved.

It is worth noting that if we can somehow replace the number $t_i$ of type $%
i $ pairs by a notion of density, we would be able to obtain the same
theorem even if we are in $\mathbb{Z}^n$ instead of $\mathbb{Z}_q^n$, i.e.,
if we drop the periodicity condition.

\begin{theorem}
\label{thm:Postextend} For any $n, e$ satisfying $n \ge 6$ and $e \ge \frac{%
\sqrt{2}}{2} n - \frac{3}{4} \sqrt{2} - \frac{1}{2}$, there is no $PL(n,e)$%
-code.
\end{theorem}

\begin{proof}
Let $C$ be a $PL(n,e)$-code in $\mathbb{Z}^n$. Denote by $B_n(N)$ the $n$%
-dimensional box $[-N, N]^n$. Fix $n$ and $e$, and let $N$ be an integer
variable that is sufficiently large. Let $t_i$ be the number of type $i$
pairs $(S, T)$ where $S$ is a Lee sphere centered at a codeword in $B_n(N)$.
Counting with respect to $S$, we immediately have 
\begin{equation*}
t_i = \lvert C \cap B_n(N) \rvert \cdot g_i.
\end{equation*}
On the other hand, if we let $t_i^\prime$ be the number of type $i$ pairs $%
(S, T)$ where $T$ is a sector contained in $B_n(N-e)$, then 
\begin{equation*}
t_i^\prime = \sum_{T \subseteq B_n(N-e)}^{} t_{T,i}.
\end{equation*}

If $(S,T)$ is a pair of type $i\geq 1$ and $T\subseteq B_{n}(N-e)$, then the
center of $S$ is in $B_{n}(N)$. Thus $t_{i}\geq t_{i}^{\prime }$ and their
difference is at most the number of pairs $(S,T)$ with $T\cap
(B_{n}(N+e)\setminus B_{n}(N-e))\neq \emptyset $. Hence 
\begin{align}
0\leq t_{i}-t_{i}^{\prime }& \leq 2^{6}\cdot \binom{n}{6}\cdot \lvert
B_{n}(N+e+6)\setminus B_{n}(N-e-6)\rvert  \notag \\
& =O(N^{n-1}).  \label{eq:Post11}
\end{align}%
Here, $\binom{n}{6}$ represents the different possible orientations of the
sectors, and we include $2^{6}$ because each $6$-dimensional sector has
non-empty intersection with at most $2^{6}$ spheres.

From Lemma~\ref{lem:Post1}, we have 
\begin{equation}
t_{1}^{\prime }-t_{7}^{\prime }-10t_{22}^{\prime }+10t_{42}^{\prime
}+t_{57}^{\prime }-t_{63}^{\prime }\geq 0.  \label{eq:Post12}
\end{equation}%
From Lemma~\ref{lem:Post2} and the fact that a negative integer is at
most $-1$, we have 
\begin{align}
& t_{1}-t_{7}-10t_{22}+10t_{42}+t_{57}-t_{63}\leq -\lvert C\cap
B_{n}(N)\rvert  \notag \\
& =-\frac{\lvert B_{n}(N)\rvert }{\lvert S(n,e)\rvert }+O(N^{n-1})=-\frac{%
2^{n}}{\lvert S(n,e)\rvert }N^{n}+O(N^{n-1})  \label{eq:Post13}
\end{align}%
by taking the sum over all codewords in $B_{n}(N)$, because
\[
B_{n}(N-e) \subset \bigcup_{\mathbf{a} \in C \cap B_n(N)} (\mathbf{a} + S(n,e)) \subset B_{n}(N+e).
\]
Equations~%
\eqref{eq:Post11}, \eqref{eq:Post12}, \eqref{eq:Post13} contradict each
other as $N$ tends to infinity.
\end{proof}

\begin{remark}
It might be possible to improve the constant $\sqrt{2}/2$ by counting
higher-dimensional sectors, or shapes other than $2\times \cdots \times
2\times 1\times \cdots \times 1$ boxes. However, it would require much more
computing, and it seems unlikely that it gives a bound that is better than
linear.
\end{remark}

\subsection{Lepist\"{o}'s Bound}

Lepist\"{o} \cite{Le} proved a bound much stronger than Post's (see Theorem %
\ref{thm:Lep}) by modifying an argument of Astola \cite{A2}.

\setcounter{tempthm}{\arabic{theorem}} \setcounter{theorem}{5}

\begin{theorem}[\protect\cite{Le}]
For any $n,e,q$ satisfying $n<(e+2)^{2}/2.1,e\geq 285$, and $q\geq 2e+1$,
there is no $PL(n,e,q)$-code.
\end{theorem}

\setcounter{theorem}{\arabic{tempthm}}

The proof is much more complicated, and thus we only outline the main idea
of the proof. Lepist\"{o} considers the set 
\begin{align*}
\Lambda (e,s)=\{\mathbf{x}\in \mathbb{Z}_{q}^{n}(\text{or }\mathbb{Z}^{n})&
:\delta _{L}(\mathbf{x},\mathbf{0})=e+2, \\
& \;-s<x_{i}\leq s\text{ for all }i\}.
\end{align*}

\begin{lemma}[{\protect\cite[Lemmas~2b, 10]{Le}}]
\label{lem:Lep2b} If $C \subseteq \Lambda(e,s)$ is an $e$-error-correcting
Lee code with at least two elements, where $q \ge 2s$ and $s \ge 2$, then
there exist two codewords with Lee distance at most 
\begin{equation*}
\frac{e+2}{\lvert C \rvert - 1} \Bigl( \lvert C \rvert \Bigl( 2 - \frac{e+2}{%
n} \Bigr) + 4s - 6 \Bigr).
\end{equation*}
In particular, this quantity is at least $2e+2$, since any two elements in $%
\Lambda(e,s)$ have even distance.
\end{lemma}

On the other hand, a standard averaging argument shows the existence of a
translate of $\Lambda (e,s)$ with many codewords.

\begin{lemma}[{\protect\cite[Lemma~1b]{Le}}]
\label{lem:Lep1b} Let $C$ be a $PL(n,e,q)$-code, where $q \ge 2s$ and $e \ge
2$. Then there exists an $\mathbf{a} \in \mathbb{Z}_q^n$ such that 
\begin{equation*}
\lvert (\mathbf{a} + \Lambda(e, s)) \cap C \rvert \ge \frac{\lvert
\Lambda(e,s) \rvert}{\lvert S(n,e,q) \rvert - \lvert S(n,e-2,q) \rvert}.
\end{equation*}
\end{lemma}

If $C$ is an $e$-error-correcting Lee code, then $(\mathbf{x} +
\Lambda(e,s)) \cap C$ is an $e$-error-correcting Lee code contained in $%
\mathbf{x} + \Lambda(e,s)$. Lepist\"{o} then uses Lemma~\ref{lem:Lep2b} to
arrive at a contradiction in the case $n < (e+2)^2 / 2.1$ and $e \ge 285$.
After Lemmas~\ref{lem:Lep2b} and \ref{lem:Lep1b}, the proof is purely about
estimating $\lvert \Lambda(e,s) \rvert$ and $\lvert S(n,e,q) \rvert$.

As in the case of Post's result, we can easily extend this to the case in $%
\mathbb{Z}^n$ by replacing counting the number by computing a density.

\begin{lemma}
\label{lem:Lep1bp} Let $C$ be a $PL(n,e)$-code, where $e \ge 2$. Then there
exists an $\mathbf{a} \in \mathbb{Z}^n$ such that 
\begin{equation*}
\lvert (\mathbf{a} + \Lambda(e,s)) \cap C \rvert \ge \frac{\lvert
\Lambda(e,s) \rvert}{\lvert S(n,e) \rvert - \lvert S(n,e-2) \rvert}.
\end{equation*}
\end{lemma}

\begin{proof}
We make a similar argument as in Theorem~\ref{thm:Postextend}. We count the
cardinality $t$ of the set 
\begin{equation*}
\{ (\mathbf{a}, \mathbf{x}) : \mathbf{x} - \mathbf{a} \in \Lambda(e,s), 
\mathbf{x} \in C \cap B_n(N) \}.
\end{equation*}
With respect to $\mathbf{x}$, we count 
\begin{equation}
t = \lvert \Lambda(e,s) \rvert \cdot \lvert C \cap B_n(N) \rvert = \frac{2^n
\lvert \Lambda(e,s) \rvert}{\lvert S(n,e) \rvert} N^n + O(N^{n-1}),
\label{eq:Lep11}
\end{equation}
as in Theorem~\ref{thm:Postextend}, because the size of $B_n(N)$ is about $%
2^n N^n$ and $C$ has density $1 / \lvert S(n,e) \rvert$.

On the other hand, we can count $t$ with respect to $\mathbf{a}$. Here, note
that if $\mathbf{a} \in C + S(n, e-2)$ or $\mathbf{a} \notin B_n(N+e+3)$
then $\mathbf{a} + \Lambda(e,s)$ and $C \cap B_n(N)$ are always disjoint.
Thus 
\begin{align}
t &= \sum_{\mathbf{a} \in B_n(N+e+3) \setminus (C + S(n,e-2))}^{} \lvert (%
\mathbf{a} + \Lambda(e,s)) \cap C \rvert  \notag \\
&\le \lvert B_n(N+e+3) \setminus (C + S(n,e-2)) \rvert  \notag \\
&\qquad \cdot \max_{\mathbf{a}} \lvert (\mathbf{a} + \Lambda(e,s)) \cap C
\rvert  \notag \\
&\le \bigl( \max_\mathbf{a} \lvert (\mathbf{a} + \Lambda(e,s)) \cap C \rvert %
\bigr)  \notag \\
&\qquad \Bigl( \frac{\lvert S(n,e) \rvert - \lvert S(n,e-2) \rvert}{\lvert
S(n,e) \rvert} 2^n N^n + O(N^{n-1}) \Bigr).  \label{eq:Lep12}
\end{align}
From Equations~\eqref{eq:Lep11} and \eqref{eq:Lep12} it immediately follows
that 
\begin{equation*}
\max_{\mathbf{a}} \lvert (\mathbf{a} + \Lambda(e,s)) \cap C \rvert \ge \frac{%
\lvert \Lambda(e,s) \rvert}{\lvert S(n,e) \rvert - \lvert S(n,e-2) \rvert}
\end{equation*}
after taking the limit $N \to \infty$.
\end{proof}

We may use Lemma~\ref{lem:Lep1bp} instead of Lemma~\ref{lem:Lep1b}. The
proof of Theorem~\ref{thm:Lep} uses more that just Lemmas~\ref{lem:Lep2b}
and \ref{lem:Lep1b}. However, other lemmas use essentially the same ideas,
and using the density trick, they can all be modified to take $PL(n,e)$%
-codes into account. Because the complicated structure of Lepist\"{o}'s
proof, it is nearly impossible to give a more detailed account of the
modification without unraveling technical details.

\begin{theorem}
\label{thm:Lepextend} For any $n, e$ satisfying $n < (e+2)^2 / 2.1$ and $e
\ge 285$, there is no $PL(n,e)$-code.
\end{theorem}

\begin{proof}
Let $C$ be a $PL(n,e)$-code. Lemma~\ref{lem:Lep1bp} implies that there
exists an $\mathbf{a} \in \mathbb{Z}^n$ such that 
\begin{equation}
\alpha = \lvert (\mathbf{a} + \Lambda(e,s)) \cap C \rvert \ge \frac{\lvert
\Lambda(e,s) \rvert}{\lvert S(n,e) \rvert - \lvert S(n,e-2) \rvert}.
\label{eq:Lepext1}
\end{equation}
Then $(\mathbf{a} + \Lambda(e,s)) \cap C$ is a $e$-error correcting Lee
code, and hence Lemma~\ref{lem:Lep2b} shows that 
\begin{equation}
2e + 2 \le \frac{e+2}{\alpha - 1} \Bigl( \alpha \Bigl( 2 - \frac{e+2}{n} %
\Bigr) + 4s - 6 \Bigr).  \label{eq:Lepext2}
\end{equation}
The two inequalities \eqref{eq:Lepext1} and \eqref{eq:Lepext2}, with Lepist%
\"{o}'s estimates give a contradiction.
\end{proof}

\subsection{Linear Programming}

We would like to sketch one more method that can be used to prove
nonexistence of $PL(n,e)$-codes for large $e$. The idea originates from
Golomb and Welch's observation that $PL(n,e)$-codes induce translational
packings of $\mathbb{R}^n$ by cross-polytopes.

It is extremely difficult to obtain an effective upper bound for the packing
density of an arbitrary tile. In the case of $n$-dimensional spheres, Cohn
and Elkies \cite{CH} developed a tool for proving upper bounds for packing
densities, which eventually determined the densest sphere packing in
dimensions $8$ \cite{Via} and $24$ \cite{CKMRV}.

\begin{theorem}[{\protect\cite[Theorem~B.1]{CH}}]
Let $V\subseteq \mathbb{R}^{n}$ be a convex body, symmetric with respect to
the origin, $f:\mathbb{R}^{n}\rightarrow \mathbb{R}$ be a nonzero function,
and $\hat{f}(\mathbf{t})$ denote its Fourier transform. Assume that:

\begin{itemize}
\item[(0)] $\lvert f\rvert $ and $\lvert \hat{f}\rvert $ decay faster than $%
\lvert \mathbf{x}\rvert ^{-n-\varepsilon }$ for some $\varepsilon >0$,

\item[(1)] $f(\mathbf{x}) \le 0$ for $\mathbf{x} \notin V$,

\item[(2)] $\hat{f}(\mathbf{t}) \ge 0$ for all $\mathbf{t}$.
\end{itemize}

\noindent Then packings of $\mathbb{R}^{n}$ by translates of $V$ have
density at most 
\begin{equation*}
\frac{\func{vol}(V)f(\mathbf{0})}{2^{n}\hat{f}(\mathbf{0})},
\end{equation*}
where $\func{vol}(V)$ stands for the volume of $V$.
\end{theorem}

Instead of applying this theorem directly to the cross-polytope, we use a
discrete analogue of the theorem and apply it to the discrete Lee sphere.
The proof is essentially the same, but using functions $\mathbb{Z}^n \to 
\mathbb{R}$ and not scaling the tile by $2$, because there is no good notion
of convexity.

\begin{theorem}
\label{thm:CEdiscrete} Let $V\subseteq \mathbb{Z}^{n}$ be a finite subset,
symmetric with respect to the origin. Suppose $f:\mathbb{Z}^{n}\rightarrow 
\mathbb{R}$ is a nonzero function such that:

\begin{itemize}
\item[(0)] $\lvert f\rvert $ decay faster than $\lvert \mathbf{x}\rvert
^{-n-\varepsilon }$ for some $\varepsilon >0$,

\item[(1)] $f(\mathbf{x}) \le 0$ for $\mathbf{x} \notin V + V$,

\item[(2)] $\hat{f}(\mathbf{t}) = \sum_{\mathbf{x} \in \mathbb{Z}^n}^{} f(%
\mathbf{x}) e^{-2 \pi i \mathbf{x} \cdot \mathbf{t}} \ge 0$ for all $\mathbf{%
t}$.
\end{itemize}

\noindent Then packings of $\mathbb{Z}^{n}$ (not of $\mathbb{R}^{n}$) by
translates of $V$ have density at most 
\begin{equation*}
\frac{\lvert V\rvert f(\mathbf{0})}{\hat{f}(\mathbf{0})}.
\end{equation*}
\end{theorem}

In particular, if $\lvert V\rvert f(\mathbf{0})<\hat{f}(\mathbf{0})$ then $V$
does not tile $\mathbb{Z}^{n}$ by translations. It is noteworthy that if $%
f=\chi _{V}\ast \chi _{V}$, where $\chi _{V}$ is the characteristic function
of $V$ and $\ast$ denotes the convolution, then all conditions are satisfied
and $\lvert V\rvert f(\mathbf{0})/\hat{f}(\mathbf{0})=1$.

Checking whether there exists such a function $f$ with $\lvert V \rvert f(%
\mathbf{0}) < \hat{f}(\mathbf{0})$ is now a linear programming problem. But
in its current form, the problem is not easily computable since the
restrictions are complicated. Thus we consider one special situation in
which the conditions become much simpler.

Denote $\tilde{g}(\mathbf{x}) = g(-\mathbf{x})$. For a function $g$ with
fast decay, we use a linear perturbation $f = (\chi_V - \epsilon g) \ast
(\chi_V - \epsilon \tilde{g})$ for $0 < \epsilon \ll 1$. Then conditions~(0)
and (2) are automatically satisfied, and (1) also is satisfied up to first
order of $\epsilon$ if and only if $(\chi_V \ast \tilde{g})(\mathbf{x}) \ge
0 $ for all $\mathbf{x} \notin V + V$. We then obtain the following
corollary.

\begin{corollary}
\label{cor:CEdiscretelinear} If there exists a function $g : \mathbb{Z}^n
\to \mathbb{R}$ satisfying the following conditions, then there is no $%
PL(n,e)$-code:

\begin{itemize}
\item[(0)] $\lvert g\rvert $ decays faster than $\lvert \mathbf{x}\rvert
^{-n-\varepsilon }$ for some $\varepsilon >0$,

\item[(1)] $g(\mathbf{x}) = 0$ for $\mathbf{x} \in S(n,e)$,

\item[(2)] $(g \ast \chi_{S(n,e)})(\mathbf{x}) \ge 0$ for $\mathbf{x} \notin
S(n,2e)$,

\item[(3)] $\sum_{\mathbf{x} \in \mathbb{Z}^n}^{} g(\mathbf{x}) < 0$. (The
sum converges absolutely by (0).)
\end{itemize}
\end{corollary}

We note that this corollary has a separate elementary proof that does not
appeal to Theorem~\ref{thm:CEdiscrete}. But it is clear that, while
difficult to use, Theorem~\ref{thm:CEdiscrete} is much stronger than
Corollary~\ref{cor:CEdiscretelinear} by itself.

Nevertheless, Corollary~\ref{cor:CEdiscretelinear} immediately yields a
significant bound. Denote by $G$ the isometry group of $\mathbb{Z}^n$,
generated by permutations of axes and reflections and consisting of $2^n
\cdot n!$ elements. We also introduce the notation 
\begin{align*}
(m_1^{\alpha_1} &, m_2^{\alpha_2}, \ldots, m_k^{\alpha_k}) \\
&= ( \overbrace{m_1, \ldots, m_1}^{\alpha_1}, m_2, \ldots, \overbrace{m_k,
\ldots, m_k}^{\alpha_k}, 0, \ldots, 0) \in \mathbb{Z}^n,
\end{align*}
which makes sense for $\sum_{i}^{} \alpha_i \le n$. The following
proposition can be proven by explicit calculation.

\begin{proposition}
Assume that $e\geq 1$ and $n\geq 2e+2$. Define the function $h:\mathbb{Z}%
^{n}\rightarrow \mathbb{R}$ as 
\begin{equation*}
h(\mathbf{x})=%
\begin{cases}
-1 & \mathbf{x}=(1^{e+1}), \\ 
\frac{4(n-e-2)(n-e-1)}{e(e+3)(2n-3e-3)} & \mathbf{x}=(1^{e+3}), \\ 
\frac{4(n-e-1)}{(e+1)(2n-3e-3)} & \mathbf{x}=(1^{e+1},2^{1}), \\ 
\frac{(2e+1)(e+1)}{2e(n-2e-1)} & \mathbf{x}=(1^{e},3^{1}), \\ 
0 & \text{otherwise}.%
\end{cases}%
\end{equation*}%
Then the function $g:\mathbb{Z}^{n}\rightarrow \mathbb{R}$ defined by 
\begin{equation*}
g(\mathbf{x})=\frac{1}{\lvert G\rvert }\sum_{\gamma \in G}^{{}}h(\gamma
\cdot \mathbf{x})
\end{equation*}%
satisfies conditions~(0), (1), and (2) of Corollary~\ref%
{cor:CEdiscretelinear}.
\end{proposition}

Note that $\sum_{\mathbf{x}\in \mathbb{Z}^{n}}^{{}}g(\mathbf{x})=\sum_{%
\mathbf{x}\in \mathbb{Z}^{n}}^{{}}h(\mathbf{x})$. Thus if $e\geq 1$ and $%
n\geq 2e+2$ and 
\begin{equation*}
\frac{4(n-e-1)}{2n-3e-3}\Bigl(\frac{n-e-2}{e(e+3)}+\frac{1}{e+1}\Bigr)+\frac{%
(2e+1)(e+1)}{2e(n-2e-1)}<1,
\end{equation*}%
then there is no $PL(n,e)$-code. Computing the interval of $n$ for which the
inequality is satisfied gives the following corollary.

\begin{corollary}
\label{cor:LP}If $e \geq 18$ and $3e + 21 \leq n\leq \frac{1}{2}e^{2}-20$,
then there is no $PL(n,e)$-code.
\end{corollary}

For $n \leq 3e + 21$, it is likely that another choice of $g$, which takes
more care of the case when $n$ is small compared to $e$, would prove the
nonexistence of $PL(n,e)$-codes.

It is curious that this method gives almost the same bound as Lepist\"{o}'s.
We do not have a good explanation for this, but we also do not think the
method itself is equivalent to Lepist\"{o}'s. On the other hand, numerical
experiments suggest that it is unlikely that Corollary~\ref%
{cor:CEdiscretelinear} by itself, with a clever choice of function, is
powerful enough to resolve the Golomb--Welch conjecture.

\section{Algebraic Approaches to Translational Tiling Problems}

\label{sec:numbertheoretic}

As indicated in Section~\ref{sec:GW}, discussed in detail in Section~\ref%
{sec:PostLep}, and in following Section~\ref{subsec:generalization} there is
a wide variation of methods and techniques how the Golomb--Welch conjecture
has been attacked. Unfortunately, it seems to us that none of these
approaches is powerful enough to solve the conjecture.

We guess that to solve the Golomb--Welch conjecture new methods and
techniques have to be introduced, new conditions under which there exists a
(periodic/lattice) tiling of $\mathbb{Z}^{n}$ by translates of a finite set $%
V$ will have to be found. Along this line we provide in this section a
necessary condition for the existence of a tiling of $\mathbb{Z}^{n}$ by a
generic (arbitrary) tile $V$. This condition is proved by the so called
polynomial method. In the second part of this section we briefly mention
applications of Fourier analysis in tilings by translates. We state there
(without providing a proof) a condition for a tile such that all tilings by
this tile are periodic. Most likely this condition cannot be applied to the
Golomb--Welch conjecture. However we believe that further development of
Fourier analysis methods might contribute to the solution of the
Golomb--Welch conjecture in an essential way. Finally, in our quest to solve
the Golomb--Welch conjecture we have focused also on tilings by translates
of a tile of prime size. Later we looked at this types of tilings in its own
right. Hence, we have (re)proved the statement that each tiling by
translates of a tile of prime size is periodic, and also that if there is a
tiling by a tile of prime size then there is also a lattice tilings by this
tile. We believe that this statement can be further strengthen to: All
tilings by a tile of prime size are lattice one. A brief outline of the
proof of this conjecture for tiles of small size is given. We guess that it
might be possible to prove additional properties of these lattice tilings
that will show that the Golomb--Welch conjecture is true in the case when
the corresponding Lee sphere $S(n,e)$ is of prime size.

\subsection{Polynomial Method}

\label{subsec:polynomial}

We describe \emph{the Polynomial method} that has been originally introduced
by Barnes \cite{Barnes} who applied this method to tilings of a box with
bricks \cite{Barnes1}. Later, the same method has been rediscovered
independently in \cite{H-K} and \cite{KS}, where the authors focus on
Nivat's conjecture. Therefore results in \cite{H-K} and \cite{KS} overlap
only in Theorem~\ref{thm:period} (see the subsection on tiles of prime size).

Let $\mathcal{T} = \{V + \mathbf{l}: \mathbf{l} \in \mathcal{L} \}$ be a
tiling of $\mathbb{Z}^{n}$ by translates of $V$. We define a linear map $T_{%
\mathcal{T}} : \mathbb{Z} [x_{1}^{\pm 1},\dots ,x_{n}^{\pm 1}]\rightarrow 
\mathbb{Z}$, where $\mathbb{Z}[x_{1}^{\pm 1},\dots ,x_{n}^{\pm 1}]$ is the
commutative ring of Laurent polynomials generated by $x_{1}^{\pm 1},\dots
,x_{n}^{\pm 1}$, such that, for every $(a_{1},\dots ,a_{n})\in \mathbb{Z}%
^{n} $, 
\begin{equation*}
T_{\mathcal{T}} (x_{1}^{a_{1}}\cdots x_{n}^{a_{n}}) =%
\begin{cases}
1 & \text{if }(a_{1},\cdots ,a_{n})\in \mathcal{L} \\ 
0 & \text{otherwise.}%
\end{cases}%
\end{equation*}
If the tiling $\mathcal{T}$ is clear from the context we will drop the
subscript and write simply $T$. We note that $T$ is uniquely determined as
the monomials $x_{1}^{a_{1}}\cdots x_{n}^{a_{n}}$ form a basis of the ring
as a $\mathbb{Z}$-module. Let $Q_{V}\in \mathbb{Z}[x_{1}^{\pm 1},\ldots
,x_{n}^{\pm 1}]$ be a polynomial associated with $V$, 
\begin{equation*}
Q_{V}(x_{1},\dots ,x_{n})=\sum_{(a_{1},\dots ,a_{n})\in
(-V)}x_{1}^{a_{1}}\cdots x_{n}^{a_{n}}.
\end{equation*}
Then for any monomial $x_{1}^{m_{1}}\cdots x_{n}^{m_{n}}$, 
\begin{align*}
T(x_{1}^{m_{1}} &\cdots x_{n}^{m_{n}}Q_{V}) \\
& =\sum_{(a_{1},\dots ,a_{n}) \in (-V)}\lvert \{(a_{1}+m_{1},\dots
,a_{n}+m_{n})\}\cap \mathcal{L}\rvert \\
& =\lvert (-V+(m_{1},\dots ,m_{n}))\cap \mathcal{L}\rvert =1.
\end{align*}%
Since the map $T$ is linear and any polynomial is a linear combination of
monomials, we can immediately extend this equality to 
\begin{equation*}
T(PQ_{V})=P(1,\dots ,1)
\end{equation*}%
for any polynomial $P\in \mathbb{Z}[x_{1}^{\pm 1},\dots ,x_{n}^{\pm 1}]$.

In what follows we will present results proved by utilizing properties of
the linear map $T$ and the polynomial $Q_{V}$, i.e., by using the polynomial
method.

We start with a technical statement that will be used in the proof of
Conjecture \ref{conj:primetile} for tiles of small size:

\begin{theorem}
\label{C} Let $\mathcal{T}$ be a tiling of $\mathbb{Z}^{n}$ by translates of 
$V$, and let $a$ be an integer relatively prime to $\lvert V \rvert$. Then,
for any polynomial $P \in \mathbb{Z}[x_{1}^{\pm 1},\dots, x_{n}^{\pm 1}]$,
we have 
\begin{equation*}
T(PQ_{V}(x_{1}^{a},\dots ,x_{n}^{a}))=P(1,\dots ,1).
\end{equation*}
\end{theorem}

\begin{proof}
We start with the case $a>0.$ Since the map $T$ is linear, we only need to
prove $T(MQ(x_{1}^{a},\dots ,x_{n}^{a}))=1$ for any monomial $M$. To see
this it suffices to show $T(MQ(x_{1}^{a},\dots ,x_{n}^{a}))\equiv 1(\func{mod%
}a).$ Indeed, we have 
\begin{align}
T(MQ_{V}& (x_{1}^{a},\dots ,x_{n}^{a})Q_{V})  \notag \\
& =\sum_{\mathbf{v}\in (-V)}T(M\cdot x_{1}^{v_{1}}\cdots x_{n}^{v_{n}}\cdot
Q_{V}(x_{1}^{a},\dots ,x_{n}^{a}))  \notag \\
& \geq \sum_{\mathbf{v}\in (-V)}1=\lvert V\rvert,  \label{eqn:ineq1}
\end{align}%
because the map $T$ takes polynomials with nonnegative coefficients to
nonnegative values, $T(MQ_{V}(x_{1}^{p},\dots ,x_{n}^{p}))\geq 1$ for all
monomials $M$. On the other hand, 
\begin{equation*}
T(MQ_{V}(x_{1}^{a},\dots ,x_{n}^{a})Q_{V})=Q_{V}(1^{a},\dots ,1^{a})=\lvert
V\rvert.
\end{equation*}%
It follows that the equality holds for every term in \eqref{eqn:ineq1}. For
some fixed $\mathbf{v}\in (-V)$, we have $T(M\cdot x_{1}^{v_{1}}\cdots
x_{n}^{v_{n}}\cdot Q(x_{1}^{a},\dots ,x_{n}^{a}))=1$ for every monomial $M$.
Therefore $T(MQ(x_{1}^{a},\dots ,x_{n}^{a}))=1$ for every monomial $M.$

The congruence $T(MQ(x_{1}^{a},\dots ,x_{n}^{a}))\equiv 1(\func{mod}a)$ will
be proved by induction on the total number $k$ of prime factors of $a$. As
noted above, $T(PQ_{V})=P(1,\dots ,1)$ for any polynomial $P\in \mathbb{Z}%
[x_{1}^{\pm 1},\dots ,x_{n}^{\pm 1}]$. Let $a=p,$ where $p$ is a prime. Then 
\begin{align*}
T(MQ_{V}(x_{1}^{p},& \dots ,x_{n}^{p}))\equiv
T(MQ_{V}^{p})=T(MQ_{V}^{p-1}Q_{V}) \\
& =(Q_{V}(1,\dots ,1))^{p-1}=\lvert V\rvert ^{p-1}\equiv 1(\func{mod}p).
\end{align*}%
\noindent since $T(RQ_{V})=R(1,\dots ,1)$ for any polynomial $R$. For $k>1,$
let $q$ is a prime factor of $a.$ By induction hypothesis we have $%
T(M(Q_{V}(x_{1}^{\frac{a}{q}},\dots ,x_{n}^{\frac{a}{q}}))\equiv 1(\func{mod}%
\frac{a}{q})$ which in turn implies $T(P(Q_{V}(x_{1}^{\frac{a}{q}},\dots
,x_{n}^{\frac{a}{q}}))=P(1,...,1)$ for any polynomial $P\in \mathbb{Z}%
[x_{1}^{\pm 1},\dots ,x_{n}^{\pm 1}]$. Hence,%
\begin{align*}
T(MQ_{V}(x_{1}^{a},& \dots ,x_{n}^{a}))\equiv T(MQ_{V}^{q}((x_{1}^{\frac{a}{q%
}},\dots ,x_{n}^{\frac{a}{q}}))\\
&=T(MQ_{V}^{q-1}(x_{1}^{\frac{a}{q}},\ldots
,x_{n}^{\frac{a}{q}})Q_{V}(x_{1}^{\frac{a}{q}},\dots ,x_{n}^{\frac{a}{q}}))
\\
& =(Q_{V}(1,\dots ,1))^{q-1}=\lvert V\rvert ^{q-1}\equiv 1(\func{mod}q)
\end{align*}
for any prime factor $q$ of $a$. Now $T(MQ_{V}(x_{1}^{a},\dots
,x_{n}^{a}))\equiv 1(\func{mod}a)$ follows from the fact that if $F\equiv 1(%
\func{mod}q)$ then $F\equiv 1(\func{mod}q^{t})$ for any $t\in N,$ and from
the Chinese Reminder Theorem.

To finish the proof we need to show that, for any $a>0,$ it is%
\begin{equation*}
T(PQ_{V}(x_{1}^{-a},\dots ,x_{n}^{-a}))=P(1,\dots ,1)
\end{equation*}%
for any polynomial $P\in \mathbb{Z}[x_{1}^{\pm 1},\dots ,x_{n}^{\pm 1}]$.
Again, it is sufficient to prove it for monomials. We first show 
\begin{equation*}
T(MQ_{V}(x_{1}^{-a},\dots ,x_{n}^{-a}))\leq 1
\end{equation*}%
for any monomial $M$. Suppose that 
\begin{equation*}
T(Mx_{1}^{-av_{1}}\cdots x_{n}^{-av_{n}})=T(Mx_{1}^{-au_{1}}\cdots
x_{n}^{-au_{n}})=1
\end{equation*}%
for some distinct $\mathbf{v},\mathbf{u}\in (-V)$. Then letting $M^{\prime
}=Mx_{1}^{a(-v_{1}-u_{1)}}\cdots x_{n}^{a(-v_{n}-u_{n})}$, we get 
\begin{align*}
T(M^{\prime }Q_{V})&\geq T(M^{\prime }x_{1}^{-av_{1}}\cdots
x_{n}^{-av_{n}})\\
&\qquad +T(M^{\prime }x_{1}^{-au_{1}}\cdots x_{n}^{-au_{n}})=2
\end{align*}%
which contradicts the original property of $Q_{V}$. Thus $%
T(MQ_{V}(x_{1}^{-a},\dots ,x_{n}^{-a}))\leq 1$ for all $M$. \bigskip

Consider the polynomial $MQ_{V}(x_{1}^{-a},\dots ,x_{n}^{-a})Q_{V}$. Because 
$T(MQ_{V}(x_{1}^{-a},\dots ,x_{n}^{-a})Q_{V})=Q_{V}(1,\dots ,1)=\lvert
V\rvert $ and 
\begin{equation*}
T(MQ_{V}(x_{1}^{-a},\dots ,x_{n}^{-a})Q_{V})\leq \sum_{\mathbf{v}\in
V}1=\lvert V\rvert ,
\end{equation*}%
all terms must attain equality. It follows that $T(MQ_{V}(x_{1}^{-a},\dots
,x_{n}^{-a}))=1$ for any monomial $M$.
\end{proof}

As an immediate consequence we get:

\begin{corollary}
\label{blowout} Let $\mathcal{T}=\{V+\mathbf{l}:\mathbf{l}\in \mathcal{L}\}$
be a tiling of $\mathbb{Z}^{n}$ by translates of $V$, and let $a$ be an
integer relatively prime to $\lvert V\rvert $. Then $\mathcal{T}_{a}=\{aV+%
\mathbf{l}:\mathbf{l}\in \mathcal{L}\}$ is a tiling of $\mathbb{Z}^{n}$ by
translates of a \textquotedblleft {}blowout\textquotedblright {} tile $aV=\{a%
\mathbf{v}:\mathbf{v}\in V\}$.
\end{corollary}

\begin{proof}
Set $S=aV$. Then 
\begin{align*}
Q_{S}(x_{1},\ldots ,x_{n})& =\sum_{(v_{1},\ldots ,v_{n})\in
(-V)}x_{1}^{av_{1}}\cdots x_{n}^{av_{n}} \\
& =Q_{V}(x_{1}^{a},\ldots ,x_{n}^{a}).
\end{align*}%
By Theorem~\ref{C}, 
\begin{equation*}
T(MQ_{S})=T(MQ_{V}(x_{1}^{a},\dots ,x_{n}^{a}))=M(1,\dots ,1)=1
\end{equation*}%
for any monomial $M$. Thus, for any $x\in \mathbb{Z}^{n}$, 
\begin{equation*}
\lvert (-S+x)\cap \mathcal{L}\rvert =1,
\end{equation*}%
that is, $\mathcal{T}_{a}=\{aV+\mathbf{l}:\mathbf{l}\in \mathcal{L}\}$ is a
tiling of $\mathbb{Z}^{n}$ by translates of $aV$.
\end{proof}

The following corollary can be found in \cite{Szegedy}. We provide here a
short proof of this result.

\begin{corollary}
(\cite{Szegedy}) \label{CC} Let $\mathcal{T}=\{V+\mathbf{l}:\mathbf{l}\in 
\mathcal{L}\}$ be a tiling of $\mathbb{Z}^{n}$ by translates of $V$, and let 
$a$ be an integer relatively prime to $\lvert V\rvert $. Then $\mathbf{l}+a(%
\mathbf{v}-\mathbf{w})\notin \mathcal{L}$ for each $\mathbf{l}\in \mathcal{L}
$ and $\mathbf{v}\neq \mathbf{w}\in V$.
\end{corollary}

\begin{proof}
By Corollary~\ref{blowout}, $\mathcal{T}_{a} = \{ aV+\mathbf{l} : \mathbf{l}
\in \mathcal{L}\}$ is a tiling of $\mathbb{Z}^{n}$ by translates of $aV$,
hence $\mathbb{Z}^{n} = aV + \mathcal{L}$. Assume that $\mathbf{l} + a(%
\mathbf{v} - \mathbf{w}) \in \mathcal{L}$. Then 
\begin{align*}
\mathbf{l} + a\mathbf{v} &= a\mathbf{w} + [\mathbf{l} + a(\mathbf{v} - 
\mathbf{w})] \quad \text{but also} \\
\mathbf{l} + a\mathbf{v} &= a\mathbf{v} + \mathbf{l}.
\end{align*}
That is, $\mathbf{l} + a\mathbf{v} \in \mathbb{Z}^{n}$ would be covered by
two distinct tiles of $\mathcal{T}_{a}$.
\end{proof}

To start building a theory of tilings of $\mathbb{Z}^{n}$ by translates of a
finite tile, and to further exhibit the strength of this method, at the end
of this section we provide a necessary condition for the existence of a
tiling of $\mathbb{Z}^{n}$ by translates of a generic (arbitrary) tile $V$.

We start by recalling a famous theorem of Hilbert that will be applied in
the proof of this condition.

\begin{theorem}[Nullstellensatz]
Let $J$ be an ideal in $\mathbb{C}[x_{1}, \ldots, x_{n}]$, and $S \subseteq 
\mathbb{C}^{n}$. Denote by $\mathcal{V}(J)$ the set of all common zeros of
polynomials in $J$, and by $\mathcal{I}(S)$ the set of all polynomials in $%
\mathbb{C}[x_{1}, \ldots, x_{n}]$ that vanish at all elements of $S$. Then 
\begin{align*}
\mathcal{I}(\mathcal{V}(J)) &= \sqrt{J} \\
&= \{f \in \mathbb{C}[x_{1}, \ldots, x_{n}] : f^{n} \in J \text{ for some }
n \geq 1 \}.
\end{align*}
\end{theorem}

We can directly apply Hilbert's Nullstellensatz to prove a Laurent
polynomial version of Nullstellensatz.

\begin{lemma}
\label{lem:LaurentNulls} Let $\{f_{i}\}_{i \in I} \subseteq \mathbb{C}
[x_{1}^{\pm 1},\dots,x_{n}^{\pm 1}]$ be a set of Laurent polynomials such
that there exists no $(x_{1}, \ldots ,x_{n}) \in (\mathbb{C} \setminus
\{0\})^{n}$ with $f_{i}(x_{1}, \ldots, x_{n}) = 0$ simultaneously for all $i
\in I$. Then there exist Laurent polynomials $p_{1}, \ldots, p_{k}$ and
indices $i_{1}, \ldots, i_{k} \in I$ such that 
\begin{equation*}
f_{i_{1}} p_{1} + \cdots + f_{i_{k}} p_{k}=1.
\end{equation*}
\end{lemma}

\begin{proof}
For each $i \in I$, consider a sufficiently large positive integer $n_{i}$
which makes $(x_{1} \cdots x_{n})^{n_{i}-1} f_{i} \in \mathbb{C}[x_{1},
\ldots, x_{n}]$. Then $g_{i} = (x_{1} \cdots x_{n})^{n_{i}} f_{i}$ is not
only a polynomial, but also a multiple of $x_{1} \cdots x_{n}$. Consider the
ideal $J \subseteq \mathbb{C}[x_{1}, \ldots, x_{n}]$ generated by the
polynomials $g_{i}$. By the condition, there is no $\mathbf{x} \in (\mathbb{C%
}\setminus \{0\})^{n}$ that makes $g_{i}(\mathbf{x}) = 0$ for all $i \in I$.
On the other hand, $g_{i}(\mathbf{x})=0$ if any one of $x_{1}, \ldots, x_{n}$
is zero since the polynomial is a multiple of $x_{1} \cdots x_{n}$. Thus it
follows that 
\begin{equation*}
\mathcal{V}(J) = \{(x_{1}, \ldots, x_{n}) \in \mathbb{C}^{n} : x_{1} x_{2}
\cdots x_{n} = 0\}.
\end{equation*}
By Hilbert's Nullstellensatz, $x_{1} \cdots x_{n} \in \mathcal{I}(\mathcal{V}%
(J)) = \sqrt{J}$, i.e., there exists a positive integer $m$ for which $%
(x_{1} \cdots x_{n})^{m} \in J$.

Let $q_{1}, \ldots, q_{k}$ and $i_{1}, \ldots, i_{k}$ be the polynomials and
indices which make 
\begin{align*}
(x_{1} &\cdots x_{n})^{m} = g_{i_{1}} q_{1} + \cdots + g_{i_{k}} q_{k} \\
&= (x_{1} \cdots x_{n})^{n_{i_{1}}} f_{i_{1}} q_{1} + \cdots + (x_{1} \cdots
x_{n})^{n_{i_{k}}} f_{i_{k}} q_{k}.
\end{align*}
Then dividing both sides by $(x_{1} \cdots x_{n})^{m}$, we get 
\begin{equation*}
1 = f_{i_{1}} \frac{q_{1}}{(x_{1} \cdots x_{n})^{m-{n_{i_{1}}}}} + \cdots +
f_{i_{k}} \frac{q_{k}}{(x_{1} \cdots x_{n})^{m-{n_{i_{k}}}}}. \qedhere
\end{equation*}
\end{proof}

The following statement is the main theorem of this subsection.

\begin{theorem}
\label{D} Let $V \subset \mathbb{Z}^{n}$ be a tile with at least $2$
elements. Then there is a tiling of $\mathbb{Z}^{n}$ by translates of $V$
only if there exists $(x_{1}, \ldots, x_{n}) \in (\mathbb{C} \setminus
\{0\})^{n}$ such that $Q_{V}(x_{1}^{a}, \ldots, x_{n}^{a}) = 0$
simultaneously for all $a$ relatively prime to $\lvert V \rvert$.
\end{theorem}

\begin{proof}
Assume that there is no $(x_{1}, \ldots, x_{n}) \in (\mathbb{C}\setminus
\{0\})^{n}$ such that $Q_{V}(x_{1}^{a}, \ldots, x_{n}^{a}) = 0$
simultaneously for all $a$ relatively prime to $\lvert V \rvert$. By Lemma~%
\ref{lem:LaurentNulls}, we obtain Laurent polynomials $P_{1}, \ldots, P_{t}$
and integers $a_{1}, \ldots, a_{t}$ relatively prime with $\lvert V \rvert$
for which 
\begin{equation}
P_{1} Q(x_{1}^{a_{1}}, \ldots, x_{n}^{a_{1}}) + \cdots + P_{t}
Q(x_{1}^{a_{t}}, \ldots, x_{n}^{a_{t}}) = 1.  \label{aa}
\end{equation}
Replacing all $x_{1}, \ldots, x_{n}$ with $1$, we get 
\begin{equation}
P_{1}(1, \ldots, 1) + \cdots + P_{t}(1, \ldots, 1) = 1 / \lvert V \rvert.
\label{a}
\end{equation}

Suppose that there exists a tiling of $\mathbb{Z}^{n}$ by translates of $V$.
By \eqref{aa} and \eqref{a}, for any monomial $M$, 
\begin{align*}
&T(M) \\
& =T(M(P_{1}Q(x_{1}^{a_{1}},\ldots ,x_{n}^{a_{1}})+\cdots
+P_{t}Q(x_{1}^{a_{t}},\ldots ,x_{n}^{a_{t}}))) \\
& =T(MP_{1}Q(x_{1}^{a_{1}},\ldots ,x_{n}^{a_{1}}))+\cdots
+T(MP_{t}Q(x_{1}^{a_{t}},\ldots ,x_{n}^{a_{t}})) \\
& =P_{1}(1,\ldots 1)+\cdots +P_{t}(1,\ldots ,1) \quad \text{(by Theorem~\ref%
{C})} \\
& =1/\lvert V\rvert.
\end{align*}%
Because this is not an integer, as $\lvert V\rvert \geq 2 $, we arrive at a
contradiction.
\end{proof}

Finally we note that it can be proved that there exists a common zero $%
(x_{1},\ldots ,x_{n})\in (\mathbb{C}\setminus \{0\})^{n}$ to $%
Q_{V}(x_{1}^{a},\ldots ,x_{n}^{a})=0$ for $\gcd (a,\lvert V\rvert )=1$ if
and only if there is a common zero $(x_{1},\ldots ,x_{n})$ with $\lvert
x_{i}\rvert =1$ for all $i$. Therefore we have a slightly stronger statement.

\begin{theorem}
Let $V\subseteq \mathbb{Z}^{n}$ be a tile with at least $2$ elements. There
exists a tiling of $\mathbb{Z}^{n}$ by translates of $V$ only if there
exists a $(x_{1},\ldots ,x_{n})\in \mathbb{C}^{n}$ with $\lvert x_i \rvert
=1 $ for all $i$, such that $Q_{V}(x_{1}^{a},\ldots ,x_{n}^{a})=0$ for all $%
\gcd (a,\lvert V\rvert )=1$.
\end{theorem}

\begin{proof}
As the proof is tedious, we only provide a sketch. We need only to prove
that if there is a common solution $\mathbf{x} \in (\mathbb{C} \setminus
\{0\})^n$ to $Q_V(x_1^a, \ldots, x_n^a) = 0$, then there is also a common
solution with $\lvert x_i \rvert = 1$ for all $i$. If we write 
\begin{equation*}
Q_V(x_1, \ldots, x_n) = m_1 + m_2 + \cdots + m_{\lvert V \rvert},
\end{equation*}
where $m_i$ are monomials in $x_1, \ldots, x_n$, then we can also write 
\begin{equation*}
Q_V(x_1^a, \ldots, x_n^a) = m_1^a + m_2^a + \cdots + m_{\lvert V \rvert}^a.
\end{equation*}
Because this is zero for all $a = k \lvert V \rvert + 1$, 
\begin{equation*}
m_1 (m_1^{\lvert V \rvert})^k + \cdots + m_{\lvert V \rvert} (m_{\lvert V
\rvert}^{\lvert V \rvert})^k = 0
\end{equation*}
for all $k \in \mathbb{Z}$. Thus if we group $m_i$ by its $\lvert V \rvert$%
-th powers, the powers of $m_i$ contained in a single group shall add up to $%
0$. This means that if we replace $m_i$ with $m_i / \lvert m_i \rvert$,
their powers still add up to zero. Therefore if $(x_1, \ldots, x_n)$ is a
common solution, then 
\begin{equation*}
\Bigl( \frac{x_1}{\lvert x_1 \rvert}, \ldots, \frac{x_n}{\lvert x_n \rvert} %
\Bigr)
\end{equation*}
is also a common solution.
\end{proof}

\begin{example}
To demonstrate that the above condition is only a necessary one, consider
the Lee sphere $V=S(3,2)$. We know that there is no $PL(3,2)$-code, i.e., no
tiling of $\mathbb{Z}^{3}$ by $S(3,2)$. However there is a common root of $%
Q_{V}(x^{a},y^{a},z^{a})=0$ for all $5\nmid a$; take $x=1$, $y=e^{2\pi i/5}$%
, and $z=e^{4\pi i/5}$ for example.
\end{example}

One of the main strength of the above theorem is that it is not limited by a
special size or by a special shape of the tile. On the other hand, it is
difficult to see whether the system has a common root except in special
cases. We will see, Theorem~\ref{thm:primelattice}, that the conditions
simplifies if the size of the tile is prime. If it is composite, the
condition is hard to interpret. Therefore, it will require additional
research to enable one to apply this theorem toward the Golomb--Welch
conjecture.

\subsection{Fourier Analysis in Tilings}

\label{subsec:Fourier}

To our best knowledge Fourier analysis has been used first time in the area
of tilings by translates by Lagarias and Wang \cite{Lagarias}, and then by
Kolountzakis and Lagarias \cite{K-L}. In both these papers a tiling of the
line by a function is studied. We note for the interested reader that an
introduction to the application of Fourier analysis in tilings has been
given in \cite{Kol}. Using methods described in \cite{Kol} we have found a
sufficient condition for a generic (arbitrary) tile $V$ such that each
tiling of $Z^{n}$ by $V$ is periodic.

\begin{theorem}
\label{thm:Fourier}Let $V$ be a tile. Suppose there exist only finitely many 
$(z_{1},\ldots ,z_{n})$ with $\lvert \mathbf{z}\rvert =1$ that satisfy 
\begin{equation*}
Q_{V}(z_{1}^{k},\ldots ,z_{n}^{k})=0
\end{equation*}%
simultaneously for all $k$ with $\gcd (k,\lvert V\rvert )=1$. Then every
tiling by $V$ is periodic.
\end{theorem}

A proof of this result is rather involved. It is a part of a manuscript
where we describe our results on translational tilings obtained by Fourier
analysis \cite{H-K2}. The above theorem illustrates possibilities how
Fourier analysis can contribute to a solution of the Golomb--Welch conjecture.

\subsection{Tiles of Prime Size}

\label{subsec:primetile}

As we have seen in Section~\ref{sec:numbertheoretic}-A, algebraic properties
of the tile $V$ place interesting restrictions on the translational tiling.
In this subsection, we focus on a particularly restrictive case, when the
size of the tile is prime. As we have already seen, the conclusion of
Corollary~\ref{blowout} holds for all integers $a$ relatively prime to $%
\lvert V\rvert $. This suggests that there are more restrictions on the
tiling when $\lvert V\rvert $ has fewer prime divisors. But in the extreme
case, when $\lvert V\rvert $ is a prime number, the situation becomes more
interesting. For more details of the results presented in the current
section, we refer the reader to \cite{H-K}.

The following theorem has been first proved by Szegedy \cite{Szegedy}.
However, using the language of polynomials makes the proof more natural. We
note that an identical proof can be found in \cite{KS}.

\begin{theorem}
(\cite{Szegedy}, \cite{KS}, \cite{H-K}) \label{thm:period} Let $V\subset 
\mathbb{Z}^{n}$ be a finite set, and $\mathcal{T}=\{V+l:l\in \mathcal{L\}}$
be a tiling of $\mathbb{Z}^{n}$ by translates of $V$. If $\lvert V\rvert =p$
is prime, then $p(\mathbf{v}-\mathbf{w})$ is a period of $\mathcal{T}$ for
any $\mathbf{v},\mathbf{w}\in V$.
\end{theorem}

\begin{proof}
For any monomial $M$, 
\begin{align*}
T(MQ_{V}(x_{1}^{p},& \ldots ,x_{n}^{p}))\equiv
T(MQ_{V}^{p})=T(MQ_{V}^{p-1}Q_{V}) \\
& =(Q_{V}(1,\ldots ,1))^{p-1}=p^{p-1}\equiv 0(\func{mod}p)
\end{align*}%
since $T(RQ_{V})=R(1,\ldots ,1)$ for any polynomial $R$. On the other hand,
by definition 
\begin{equation*}
T(MQ_{V}(x_{1}^{p},\ldots ,x_{n}^{p}))=\sum_{\mathbf{v}\in
V}^{{}}T(Mx_{1}^{-pv_{1}}\cdots x_{n}^{-pv_{n}}).
\end{equation*}%
Since the sum of $\lvert V\rvert =p$ numbers, each of which is either $0$ or 
$1$, is a multiple of $p$, we conclude that the numbers are either all $0$
or all $1$. Hence for any $\mathbf{v}=(v_{1},...,v_{n})$ and $\mathbf{w}%
=(w_{1},...,w_{n})$ in $V$%
\begin{equation*}
T(Mx_{1}^{-pv_{1}}\cdots x_{n}^{-pv_{n}})=T(Mx_{1}^{-pw_{1}}\cdots
x_{n}^{-pw_{n}})=0\text{ or }1.
\end{equation*}

It follows that, for any $\mathbf{x\in }\mathbb{Z}^{n},$ the point $\mathbf{x%
}$ is in $\mathcal{L}$ if and only if $\mathbf{x+p(v-w)}$ is in $\mathcal{L}$%
. Therefore, $p(\mathbf{v-w)}$ is a period of $\mathcal{T}.$
\end{proof}

This theorem already turns the problem of finding all tilings into a finite
computation problem. But, if one is interested only in checking existence of
tilings, as in the case of the Golomb--Welch conjecture, the problem becomes
much simpler. The following theorem is stated in \cite{Szegedy}.

\begin{theorem}
(\cite{Szegedy}) \label{thm:primelattice} Let $V=\{\mathbf{v}_{0}=\mathbf{0},%
\mathbf{v}_{1},\ldots ,\mathbf{v}_{p-1}\}\subset \mathbb{Z}^{n}$ be a prime
size tile, and suppose that $\mathbf{v}_{1},\ldots ,\mathbf{v}_{p-1}$
generate $\mathbb{Z}^{n}$ as an abelian group. Then there is a tiling of $%
\mathbb{Z}^{n}$ by translates of $V$ if and only if there is a lattice
tiling of $\mathbb{Z}^{n}$ by translates of $V$, i.e., there is group
homomorphism $\phi :\mathbb{Z}^{n}\rightarrow \mathbb{Z}_{p}$ that restricts
to a bijection $V\rightarrow \mathbb{Z}_{p}$.
\end{theorem}

\begin{proof}
Suppose there is any tiling of $\mathbb{Z}^{n}$ by translates of $V$. From
Theorem~\ref{D} we get a common nonzero solution to $Q_{V}(x_{1}^{a},\ldots
,x_{n}^{a})=0$ for $1\leq a\leq p-1$. Letting $m_{k}=x_{1}^{-v_{k,1}}\cdots
x_{n}^{-v_{k,n}}$, we can write $\sum_{k=0}^{p-1} m_{k}^{a}=0$ for all $1\leq
a\leq p-1$. It follows inductively that the elementary symmetric polynomials
are $\sum_{i_{1}<\cdots <i_{a}}m_{i_{1}}\cdots m_{i_{a}}=0$ for $1\leq a\leq
p-1$. That is, $(X-m_{0})\cdots (X-m_{p-1})=X^{p}-P$ for some $P\in \mathbb{C%
}$, and because $m_{0}=1$ since $\mathbf{v}_{0}=\mathbf{0}$, we further
obtain $P=1$. Hence 
\begin{equation*}
\{x_{1}^{-v_{k,1}}\cdots x_{n}^{-v_{k,n}}\}_{0\leq k<p}=\{1,e^{2\pi
i/p},\ldots ,e^{2\pi i(p-1)/p}\}.
\end{equation*}

Because $\{ \mathbf{v}_1, \ldots, \mathbf{v}_{p-1} \}$ generate $\mathbb{Z}%
^n $, all $x_1, \ldots, x_n$ have to be powers of $e^{2\pi i/p}$. If we
write $x_k = e^{2 \pi i \alpha_k/p}$, then the homomorphism 
\begin{equation*}
\mathbb{Z}^n \to \mathbb{Z}_p; \quad (y_1, \ldots, y_n) \mapsto \alpha_1 y_1 +
\alpha_2 y_2 + \cdots + \alpha_n y_n
\end{equation*}
restricts to a bijection $V \to \mathbb{Z}_p$.
\end{proof}

Because the Lee sphere $S(n, e)$ always contains $\mathbf{0}$ and generates $%
\mathbb{Z}^n$, we have:

\begin{corollary}
Suppose $n, e \ge 1$ and $\lvert S(n,e) \rvert = p$ is prime. Then every $%
PL(n,e)$-code is periodic with period $p$ in every direction. Moreover,
there is a $PL(n,e)$-code if and only if there is a linear $PL(n,e)$-code.
\end{corollary}

\noindent Thus, in this case, the task of proving the Golomb--Welch
conjecture reduces to verifying that there is no homomorphism $\mathbb{Z}%
^{n}\rightarrow \mathbb{Z}_{p}$ that restrict to a bijection $%
S(n,e)\rightarrow \mathbb{Z}_{p}$. This was used in \cite{Kim} to prove
nonexistence of $PL(n,2)$-codes for special $n$. The primality of $\lvert
S(n,e) \rvert$ heavily depends on $n$ and $e$. It is very likely that $%
\lvert S(n,2) \rvert = 2n^2 + 2n + 1$ is prime for infinitely many $n$,
while $\lvert S(n,3)\rvert =(2n+1)(2n^{2}+2n+3)/3$ is never prime for $n\geq
2$.

The restrictiveness of tilings by prime size tiles raises the natural
question: Can we classify all such tilings in some sense? This is not a
question directly related to the Golomb--Welch conjecture. However we think
this illustrates well the strength of using polynomials in tiling problems.

\begin{conjecture}
\label{conj:primetile} Let $V=\{\mathbf{0},\mathbf{v}_{1},\ldots ,\mathbf{v}%
_{p-1}\}$ be a prime size tile, and suppose that $\mathbf{v}_{1},\ldots ,%
\mathbf{v}_{p-1}$ generate $\mathbb{Z}^{n}$ as an abelian group. Then all
tilings of $\mathbb{Z}^{n}$ by translates $V$ are lattice.
\end{conjecture}

It turns out there is \textquotedblleft{}universal\textquotedblright{} tile
for this conjecture. Consider the \emph{semi-cross} $V_{p-1} = \{ \mathbf{0}%
, \mathbf{e}_1, \ldots, \mathbf{e}_{p-1} \} \subset \mathbb{Z}^{p-1}$. Given
any tile $V$ satisfying the assumptions of Conjecture~\ref{conj:primetile},
define a homomorphism $\phi : \mathbb{Z}^{p-1} \to \mathbb{Z}^n$ by $\phi(%
\mathbf{e}_i) = \mathbf{v}_i$. Observe that if $\mathcal{T} = \{ V + \mathbf{%
l} : \mathbf{l} \in \mathcal{L} \}$ is a tiling of $\mathbb{Z}^n$ by $V$,
then $\mathcal{T}_0 = \{ V_{p-1} + \mathbf{l} : \mathbf{l} \in \phi^{-1}(%
\mathcal{L}) \}$ is a tiling of $\mathbb{Z}^{p-1}$ by $V_{p-1}$. It is clear
that if $\phi^{-1}(\mathcal{L})$ is a lattice, then $\mathcal{L}$ is also a
lattice. Therefore the following conjecture, which is a special case, is
actually equivalent to Conjecture~\ref{conj:primetile}.

\begin{conjecture}
\label{conj:semicross} For $p$ a prime, any tiling of $\mathbb{Z}^{p-1}$ by
the semi-cross $V_{p-1} = \{\mathbf{0}, \mathbf{e}_1, \ldots, \mathbf{e}%
_{p-1}\}$ is lattice.
\end{conjecture}

We note that an equivalent conjecture, called Corr\'{a}di's conjecture, is
stated in \cite{Sz2} in the context of factorization of abelian groups. In
the same paper, the conjecture is also verified up to $p\leq 7$. However,
the proof given by Sz\'{a}bo relies on a computer search for large cliques
in a certain graph. Here we outline a readable proof that makes use of
polynomials. Again, more details are presented in \cite{H-K}.

To facilitate our discussion, we introduce new notions and notations, and
state several auxiliary results. Let $\mathcal{T} = \{ V_{p-1} + \mathbf{l}
: \mathbf{l} \in \mathcal{L} \}$ be a tiling of $\mathbb{Z}^{p-1}$ by
semi-crosses. We use the terminology of coding theory: The elements of $%
\mathbb{Z}^{p-1}$ will be called \emph{words} and the elements of $\mathcal{L%
}$, the centers of semi-crosses in $\mathcal{T}$, will be called \emph{%
codewords}.

By a word of \emph{type $[m_1^{\alpha_1}, \ldots, m_s^{\alpha_s}]$} we mean
a word having $\alpha_1$ coordinates equal to $m_1$, \ldots, $\alpha_s$
coordinates equal to $m_s$, the other coordinates equal to zero. Let $W, Z$
be words, and the word $Z - W$ is of type $[m_1^{\alpha_1}, \ldots,
m_s^{\alpha_s}]$. Then $Z$ will be called a word of \emph{type $%
[m_1^{\alpha_1}, \ldots, m_s^{\alpha_s}]$} with respect to $W$. Further, we
will say that a word $V$ is covered by a codeword $W$ if $V$ belongs to the
semi-cross centered at $W$. Finally, two words $A$ and $B$ \emph{coincide in 
$t$ coordinates}, if they have the same value in $t$ \underline{\emph{nonzero%
}} coordinates.

The following lemma facilitates analyzing possible configurations of
codewords. The proof is essentially computing the elementary symmetric
polynomials $e_k = \sum_{}^{} x_{i_1} \cdots x_{i_k}$ in terms of the
polynomials $s_k = \sum_{}^{} x_i^k$, and using Theorem~\ref{C}.

\begin{lemma}
\label{lem:counting} Let $\mathcal{T}$ be a tiling of $\mathbb{Z}^{p-1}$ by
semi-crosses, where $p $ is a prime. Then for each $k < p$, 
\begin{equation*}
T\left( \sum_{i_1 < \cdots < i_k}^{} x_{i_1} x_{i_2} \cdots x_{i_k} \right)
= \frac{\binom{p-1}{k} - (-1)^k}{p} + (-1)^k T(1).
\end{equation*}
In other words, if $O$ is a codeword then there are $\frac{1}{p} (\binom{p-1%
}{k} + (p-1) (-1)^k)$ codewords of type $[1^k]$, otherwise there are $\frac{1%
}{p}(\binom{p-1}{k} - (-1)^k)$ codewords of type $[1^k]$.
\end{lemma}

In fact, one can calculate the number of codewords of type $[m_1^{\alpha_1},
\ldots, m_s^{\alpha_s}]$ depending on whether $\mathbf{0} \in \mathbb{Z}%
^{p-1}$ is codeword. However we do not need the lemma in this generality.

Let us write $\mathbf{i} = (1, 1, \ldots, 1)$. For $k = p-1$ we see that a
word $\mathbf{w}$ is a codeword if and only if $\mathbf{w} + (1, \ldots, 1)$
is a codeword. For $k = 2$, we see that if $\mathbf{w}$ is a codeword then
there are $t = \frac{p-1}{2}$ codewords $\mathbf{u}_1, \ldots, \mathbf{u}_t$
of type $[1^2]$ with respect to $\mathbf{w}$, and $t$ codewords $\mathbf{u}%
_1^\prime, \ldots, \mathbf{u}_t^\prime$ of type $[-1^2]$ with respect to $%
\mathbf{w}$. Since they cannot share $1$ or $-1$ at the same place, 
\begin{equation*}
\sum_{i=1}^{t} (\mathbf{u}_i - \mathbf{w}) = \mathbf{i}, \quad
\sum_{i=1}^{t} (\mathbf{u}_i^\prime - \mathbf{w}) = -\mathbf{i}.
\end{equation*}
Let us denote $\mathcal{U}_2^+(\mathbf{w}) = \{ \mathbf{u}_1, \ldots, 
\mathbf{u}_t \}$ and $\mathcal{U}_2^-(\mathbf{w}) = \{ \mathbf{u}_1^\prime,
\ldots, \mathbf{u}_t^\prime \}$.

It turns out that it is useful to denote codewords in terms of cyclic
shifts. Define $\pi(a_1, \ldots, a_{p-1}) = (a_2, \ldots, a_{p-1}, a_1)$ and
write $\langle \mathbf{w} \rangle = \{ \mathbf{w}, \pi(\mathbf{w}), \pi^2(%
\mathbf{w}), \ldots \}$, the set of all cyclic shifts of $\mathbf{w}$.

\begin{theorem}
Conjecture~\ref{conj:primetile} is true for $p = 2, 3$.
\end{theorem}

\begin{proof}
For $p=2$, it is obvious. For $p=3$, simply note that all words of the
lattice generated by $(3,0)$ and $(1,1)$ are codewords. As the determinant
of the matrix consisting of the two vectors equals to $3,$ the proof follows.
\end{proof}

\begin{theorem}
Conjecture~\ref{conj:primetile} is true for $p = 5$.
\end{theorem}

\begin{proof}
Suppose that $\mathbf{w}$ is a codeword. There are $6$ codewords of type $%
[1^2]$, each of them covered by a codeword either of type $[1^2]$ or of type 
$[1^2, -1]$. Hence, as there are $2$ codewords of type $[1^2]$, there has to
be $4$ codewords of type $[1^2, -1]$. We denote this set by $\mathcal{U}_3^+(%
\mathbf{w})$, and likewise, the set of $4$ codewords of type $[-1^2, 1]$ by $%
\mathcal{U}_3^-(\mathbf{w})$.

Because $\mathcal{U}_{2}^{+}(\mathbf{w})$ has two words, we may assume
without loss of generality that $\mathcal{U}_{2}^{+}(\mathbf{w})=\langle
(1,0,1,0)\rangle $. Also $\mathcal{U}_{3}^{+}(\mathbf{w})$ has $4$ elements.
Thus we can again assume without loss of generality that $\mathcal{U}%
_{3}^{+}(\mathbf{w})=\langle (1,1,-1,0)\rangle $.

By casework near $\mathbf{w}$, it can be proved that if $\mathbf{a}\in 
\mathcal{U}_{3}^{+}(\mathbf{w})$ or $\mathbf{a}\in \mathcal{U}_{2}^{+}(%
\mathbf{w})$ then $\mathcal{U}_{2}^{+}(\mathbf{w}+\mathbf{a})=\mathcal{U}%
_{2}^{+}(\mathbf{w})$ and $\mathcal{U}_{3}^{+}(\mathbf{w}+\mathbf{a})=%
\mathcal{U}_{3}^{+}(\mathbf{w})$. It follows that all words of the lattice
generated by $\langle (1,0,1,0)\rangle $, $(1,1,-1,0)$ and $(5,0,0,0)$ are
codewords. This lattice has determinant $5$.
\end{proof}

Using similar ideas and methods, but considering many more cases, the
following theorem can be proved.

\begin{theorem}
Conjecture~\ref{conj:primetile} is true for $p=7$.
\end{theorem}

We note that in the proof of the above theorems we have not used explicitly
the fact that $p$ is a prime, other than in Lemma~\ref{lem:counting}. We
believe that the property distinguishing tilings by semi-crosses of prime
size from others is that of being cyclic. A tiling $\mathcal{T}=\{V+\mathbf{l%
}:\mathbf{l}\in \mathcal{L}\}$ is called \emph{cyclic} if there is
reordering of coordinates such that, for each codeword $\mathbf{l}$, 
\begin{equation*}
\mathbf{l}\in \mathcal{L}\;\Rightarrow \;\langle \mathbf{l}\rangle \subset 
\mathcal{L};
\end{equation*}%
\noindent that is, if for any codeword, also all its cyclic shifts are
codewords. Indeed, for $p>2$ a prime, the only lattice tiling (up to
permutation of coordinates) of $\mathbb{Z}^{p-1}$ by semi-crosses is 
\begin{equation*}
\mathcal{L}=\{\mathbf{l}\in \mathbb{Z}^{p-1}:p\mid l_{1}+2l_{2}+\cdots
+(p-1)l_{p-1}\},
\end{equation*}%
which is cyclic. On the other hand, if $n$ is not a prime, it can be proven
that no lattice tiling of $\mathbb{Z}^{n-1}$ by semi-crosses is cyclic.

For the sake of completeness we note that for any $n$ there is a unique, up
to a congruency, lattice tiling $\mathbb{Z}^{n-1}$ by semi-crosses. This
follows from a statement in \cite{H10}, that there is a lattice tiling of $%
\mathbb{Z}^{n}$ by a tile $V$ if and only if there is a homomorphism $\phi :%
\mathbb{Z}^{n}\rightarrow G,$ an additive group of order $\left\vert
V\right\vert ,$ such that a restriction of $\phi $ to $V$ is a bijection,
and from the symmetry of the semi-cross.

We guess that finding additional properties about the lattice tilings
(assuming Conjecture \ref{conj:primetile} is true) will enable one to prove
the Golomb--Welch conjecture in the case when $\lvert S(n,e)\rvert $ is a
prime.

\section{Research Inspired by the Golomb--Welch Conjecture}

\label{sec:further}

By Google Scholar there are 191 papers citing \cite{GW} (1970) and 83 papers
citing \cite{GW} (1968); this includes papers that cites both. In this
section we describe only a few of these papers.

It is very common in mathematics to generalize a problem in order to be able
to solve it. Also in the case of the Golomb--Welch conjecture there are
several modifications and generalizations. However, to the best of our
knowledge, so far none of these generalization has contributed to the
solution of the Golomb--Welch conjecture itself. In the first part of this
section we describe some of these generalizations, in the second we will
look at generalizations of perfect Lee codes.

\subsection{Generalizations of Perfect Lee Codes and of the Golomb--Welch
Conjecture}

\label{subsec:generalization}

We start this subsection with a strengthening of the Golomb--Welch
conjecture. As mentioned above, the Golomb--Welch conjecture has been proved
for all pairs $(n,e)$ where $3\leq n\leq 5$. In fact all pertinent results
proved a stronger statement: There is no tiling of $\mathbb{R}^{n}$ with Lee
spheres of radii at least two, even with different radii. Still a stronger
conjecture has been raised in \cite{GMP1}. For obvious reasons it has been
formulated in terms of tilings rather than Lee codes.

\begin{conjecture}[\protect\cite{GMP1}]
\label{3} For $n \ge 3$, there does not exist a tiling of $\mathbb{R}^{n}$
with Lee spheres of radius at least $1$ such that the radius of at least one
of them is at least $2$.
\end{conjecture}

\noindent In the same paper the conjecture is proved for $n=3$.

Now we focus on diameter-$d$ perfect Lee codes, which constitute a
generalization of perfect $e$-error-correcting Lee codes. Ahlswede et al.\ 
(see \cite{Als}) introduced diameter perfect codes for distance regular
graphs. Later the notion has been extended to metric spaces. Let $(M,\delta
) $ be a metric space. Then a set $C\subseteq M$ is a \emph{diameter-$d$ code%
} if $\delta (u,v)\geq d$ for any $u,v\in C$, and a set $A\subseteq M$ is an 
\emph{anticode of diameter $d$} if $\delta (u,v)\leq d$ for all $u,v\in A$.
Further, let $\mathcal{S}=\{S_{i}:i\in I\}$ be a family of subsets of an
underlying set $M$. Then a set $T\subseteq M$ is called a \emph{transversal}
of $\mathcal{S}$ if there is a bijection $f:I\rightarrow T$ so that $f(i)\in
S_{i}$. In what follows we restrict ourselves to $M=\mathbb{Z}^{n}$.

\begin{definition}
Let $C \subseteq \mathbb{Z}^{n}$. Then $C$ is a \emph{diameter-$d$ perfect
Lee code} in $\mathbb{Z}^{n}$ if $C$ is a diameter-$d$ code, and there is a
tiling $\mathcal{T}$ of $\mathbb{Z}^{n}$ by translates of the anticode of
diameter $d-1$ of maximum size such that $C$ is a transversal of $\mathcal{T}
$. The diameter-$d$ perfect Lee code in $\mathbb{Z}^{n}$ will be denoted by $%
DPL(n,d)$.
\end{definition}

\noindent Any error-correcting perfect Lee code is also a diameter perfect
Lee code. Indeed, it is easy to see that, for $d$ even, the anticode of
diameter $d$ of the maximum size is the Lee sphere $S(n,r)$ with $r=\frac{d}{%
2}$. Thus, for $d$ odd, $PL(n,d)$-codes are $DPL(n,e)$-codes where $e=\frac{%
d-1}{2}$. It was proved in \cite{Als} that, for $d$ odd, the anticode of
diameter $d$ of maximum size is the double-sphere $DS(n,e)=S(n,e)\cup
(S(n,e)+\mathbf{e}_{1})$ with $e=\frac{d-1}{2}$.

Etzion \cite{Etzion1} asks whether the Golomb--Welch conjecture can be
generalized to: Other than Minkowski's lattice \cite{Min} $DS(n,6)$, are
there $DPL(n,d)$-codes with $n\geq 3$, and $d>4$? Buzaglo and Etzion \cite%
{Buzaglo} partially proved the conjecture by showing that there is no $%
DPL(n,2r+2)$-code for $r>2n-4$ where $n>2$. Further generalization of
Etzion's conjecture is given in \cite{ADH}, where the notion of a perfect
distance-domination set in a graph is introduced. This notion generalizes
notions of perfect error-correcting codes, perfect diameter codes, perfect
codes in graphs \cite{Biggs}, and perfect dominating sets \cite{W}.

The Lee (Manhattan) metric is a special case of $l_{p}$ metric for $p=1$. We
note that the nonexistence of some perfect codes in $l_{p}$ metric, $1 \leq
p < \infty$ was shown in \cite{Zhang}.

\subsection{Quasi-Perfect Lee codes and $PL(n,1,q)$-codes}

\label{subsec:QPL}

As mentioned in Introduction an interest in perfect codes in the Lee metric
is due to their various applications. As it is widely believed that the
Golomb--Welch conjecture is true, i.e., that there are no $PL(n,e)$-codes
for $n\geq 3$ and $e>1$, codes \textquotedblleft {}close\textquotedblright
{} to perfect codes have been introduced and studied. To the best of our
knowledge, quasi-perfect Lee codes have been looked at first time in \cite%
{Bader1}. A code $C\subseteq \mathbb{Z}^{n}$ ($C\subseteq \mathbb{Z}%
_{q}^{n}) $ is called \emph{quasi-perfect} if the minimum distance of $C$ is 
$2e+1$ or $2e+2$ and each $\mathbf{x}$ in $\mathbb{Z}^{n}$ ($\mathbb{Z}%
_{q}^{n})$ is at distance at most $e+1$ from at least one codeword $\mathbf{y%
}\in C$. Quasi-perfect Lee codes in $\mathbb{Z}^{n}$ and in $\mathbb{Z}%
_{q}^{n}$ are denoted $QPL(n,e)$ and $QPL(n,e,q)$. In \cite{Bader1} $%
QPL(2,e,q)$-codes have been constructed for all $e>1$ and all $%
2e^{2}+2e+1\leq q<2(e+1)^{2}+2(e+1)+1.$ A fast algorithm for decoding these
codes was presented in \cite{H1}. The first $QLP(n,e)$-code with $n>2$ has
been constructed in \cite{H6}, namely it is shown there that there is $%
QPL(3,e)$-code for all $1\leq e\leq 6$. Unfortunately, it is also proved
there that for each $n$ there are only finitely many values of $e$ such that
there is a linear $QPL(n,e)$-code. Thus, the property for a code to be a
quasi-perfect code in the Lee metric is still too restrictive. The first
construction of $QPL(n,e,q)$-codes for infinitely many $n$, based on Cayley
graph, has been recently presented in \cite{Ca} and \cite{Q}. In \cite{Bibak}
it was shown that these Cayley graphs are in fact Ramanujan graphs. Another
construction of $QLP(n,e)$-codes for (possibly infinitely many) $n\equiv 1(%
\func{mod}6)$ has been provided in \cite{Zhang}, where also a construction
of quasi-perfect codes under $l_{p}$ metric is given.

$PL(n,1,q)$-codes constitute the only known class of perfect $e$%
-error-correcting codes for $n\geq 3$. Therefore, with respect to possible
applications, these codes have been looked at more closely. It is stated in 
\cite{GW} that $PL(n,1,q)$-codes might exist for $q<2n+1$ if $2n+1$ is a
perfect square; to support this claim a $PL(4,1,3)$-code is constructed
there. A complete answer to the question in the case of linear (lattice)
codes is given in \cite{H3}. It is proved there that: Let $%
2n+1=p_{1}^{\alpha _{1}}\cdots p_{k}^{a_{k}}$ be the prime factorization of $%
2n+1$ and let $p=\prod_{i=1}^{k}p_{i}$. Then a linear $PL(n,1,q)$-code
exists if and only if $p\mid q$. In particular, the smallest $q$, for which
there exists a linear $PL(n,1,q)$-code, equals $p$.

Szabo \cite{Sz1} showed that if $2n+1$ is not a prime then there exists a
non-linear but periodic $PL(n,1)$-code. Therefore, in this case, there
exists a non-linear $PL(n,1,q)$-code for suitable values of $q$; a
characterization of such $q$'s has not be given yet. If $2n+1$ is a prime
then there is the following conjecture.

\begin{conjecture}[\protect\cite{H3}]
If $2n+1$ is a prime then each $PL(n,1,q)$-code is linear, and it is a
periodic extension of the unique, up to a congruence, $PL(n,1,2n+1)$-code.
\end{conjecture}

This conjecture has been proved in \cite{H} for $n=2,3$ and in \cite{H7} for 
$n=5$. Finally, we note that non-periodic $PL(n,1)$-codes have been
constructed in \cite{H4}.

$PL(n,1)$-code can be obviously seen as a tiling of the Euclidian space by
crosses with arms of length one. In \cite{Buzaglo1}, crosses with arms of
length half are considered. These crosses might be scaled by two to form a
discrete shape. A tiling with this shape is also known as a perfect
dominating set. Buzaglo and Etzion prove that a tiling for such a shape
exists if and only if $n=2^{n}-1$ or $n=3^{t}-1$, where $t>0$. The authors
also show a strong connection of these tilings to binary and ternary perfect
codes in the Hamming scheme.

\section{Conclusions}

\label{sec:conclusions}

50 years ago, Golomb and Welch \cite{GW} raised a conjecture whose strong
version claims that there is no $PL(n,e)$-code for $n \geq 3$ and $e > 1$.
In spite of great effort and plenty of papers on the topic, this conjecture
is still far from being solved.

To provide a support for their conjecture, Golomb and Welch \cite{GW} show
that for $n \geq 3$ there exists $e_{n}$, $e_{n}$ not specified, such that
for any $e > e_{n}$ there is no $PL(n,e)$-code. For $3 \leq n\leq 5$, the
Golomb--Welch conjecture has been proved for all $e \geq 2$ (see \cite{GMP}, 
\cite{GMP1}, and \cite{H2}).

It seems that the most difficult case of the Golomb--Welch conjecture is
that of $e=2$. First, the case $e=2$ is a threshold case as there is a $%
PL(n,1)$-code for all $n.$ Second, in \cite{H2}, the proof of nonexistence
of $PL(n,e)$-codes for $3\leq n\leq 5$ and all $e\geq 2$ has been based on
the nonexistence of $PL(n,2)$-codes for the given $n$. So far the strongest
result in this direction is due to Kim \cite{Kim}, where non-existence of $%
PL(n,2)$ code is proved for (likely infinitely) many values of $n$. In
addition, the nonexistence of \emph{linear} $PL(n,2)$ codes for $n\leq 12$
has been proved in \cite{H6}.

As to the weak version of the Golomb--Welch conjecture, the nonexistence of 
\emph{periodic} $PL(n,e)$-codes has been proved by Post in \cite{P} for $%
n\geq 6$, $e\geq \frac{\sqrt{2}}{2}n-\frac{3}{4}\sqrt{2}-\frac{1}{2}$. This
result of Post was asymptotically improved by Lepist\"{o} \cite{Le} who
showed that there is no periodic $PL(n,e)$-code for $n<(e+2)^{2}/2.1$, $%
e\geq 285$. Further, the proof of nonexistence of $PL(n,e,q)$-codes for
specific values of $q$ (i.e. for specific periods for $PL(n,e)$-codes) can
be found in \cite{A1}.

As a main part of this paper we provided new results on the Golomb--Welch
conjecture. It is proved here that the condition \emph{periodic} can be
dropped from both, the result of Post and the result of Lepist\"{o}. In
addition, we showed (see Corollary~\ref{cor:LP}) that $PL(n,e)$-codes do not
exist for $e\geq 18$ and $3e+21\leq n\leq \frac{1}{2}e^{2}-20$.

The above given results have been proved by a variety of clever methods.
Anyway, we feel that none of them is strong enough to prove the
Golomb--Welch conjecture in its entirety.

In greater detail, Golomb and Welch prove Theorem~\ref{thm:tempGW} by using
the fact that a tiling of $\mathbb{R}^{n}$ by the sphere $S(n,e)$, $e$ large
enough, induces a packing of $\mathbb{R}^{n}$ by translates of the
cross-polytope with an arbitrarily high density smaller than $1$. On the
other hand, it is well-known that for $n\geq 3$, the cross-polytope does not
tile $\mathbb{R}^{n}$ by translations, and it can be shown that the packing
density of a bounded set that does not tile $\mathbb{R}^{n}$ is bounded away
from $1$. To get an explicit bound on $e$ one would need to have an upper
bound on the packing density of the cross-polytope. Unfortunately, this is a
very difficult question, and such density is known only for $n=3$ due to
Minkowski~\cite{Min}. We note that the idea of the proof of Theorem~\ref%
{thm:tempGW} has been used by several authors (see e.g.~\cite{H6}) for
generalizations of Lee codes.

In \cite{P}, to obtain an upper bound on $e_{n}$, Post shows the
nonexistence of periodic $PL(n,e)$ codes for $3\leq n\leq 5$ by proving an
inequality for the number of intersections of $3$-dimensional sectors with
Lee spheres. To get the nonexistence of $PL(n,e)$ codes for $n\geq 6,e\geq 
\frac{\sqrt{2}}{2}n-\frac{3}{4}\sqrt{2}-\frac{1}{2}$ Post considers $6$%
-dimensional sectors. It is likely, that dealing with sectors of dimension $%
>6$, would provide a better bound on $e$. However, the number of types how a 
$k$-dimensional sector can be covered by Lee spheres grows very fast with $k$%
; thus to get a needed inequality for the number of intersections of $k$%
-dimensional sectors with Lee spheres would be extremely difficult.

The method used by Lepist\"{o} in \cite{Le} to prove the nonexistence of $%
PL(n,e,q)$-codes for any $n,e,q$ satisfying $n<(e+2)^{2}/2.1,e\geq 285$, and 
$q\geq 2e+1$, is technically very involved. At the moment we do not see a
way how this method could be used to prove the nonexistence of $PL(n,e)$
codes for additional values of $(n,e).$

As for the presented linear programming approach, numerical experiments
suggest that it is unlikely that Corollary~\ref{cor:CEdiscretelinear} by
itself, with a clever choice of function $g$, is powerful enough to resolve
the Golomb--Welch conjecture.

We guess that the methods used in \cite{GMP}, \cite{GMP1}, and \cite{H2} to
prove the Golomb--Welch conjecture for small values of $n$, cannot be applied
for $n \geq 6$. The method of \cite{GMP} seems to be applicable only for $%
n=3 $, the method of \cite{GMP1} is likely computationally infeasible for
for $n\geq 4$, and the method applied in \cite{H2} leads for slightly bigger 
$n$ to a system of too many equations.

It is very common in mathematics to generalize a problem in order to be able
to solve it. Also in the case of the Golomb--Welch conjecture there are
several modifications and generalizations. However, to the best of our
knowledge, so far none of these generalization has contributed to the
solution of the Golomb--Welch conjecture itself. Some of these
generalization have been described in Section~\ref{sec:further}.

With respect to above stated, we guess that essentially new methods are
needed to prove the Golomb--Welch conjecture. Therefore, in Section~\ref%
{sec:numbertheoretic}, we have described two new avenues how to attack this
conjecture. Using a polynomial method introduced originally by Barnes~\cite%
{Barnes}, a necessary condition (see Theorem~\ref{D}) for the existence of a
tiling of $\mathbb{Z}^{n}$ by translates of a tile $V$ is proved in Section~%
\ref{subsec:polynomial}. We believe this is the first necessary condition
for a generic (arbitrary) tile. However, it is difficult to see whether the
system has a common root except in special cases. Therefore, it will require
additional research to enable one to apply this theorem toward the
Golomb--Welch conjecture.

In Section~\ref{subsec:Fourier}, using a Fourier analysis method introduced
by Lagarias and Wang~\cite{Lagarias}, we have found a sufficient condition
for a generic (arbitrary) tile $V$ such that each tiling of $\mathbb{Z}^{n}$ by $V$
is periodic (see Theorem~\ref{thm:Fourier}). This theorem illustrates
possibilities how Fourier analysis can contribute to a solution of the
Golomb--Welch conjecture. Therefore we plan to work on further development
of this method.

We guess that it will require a great effort to completely solve the
Golomb--Welch conjecture. Hence it would be also nice to solve the
conjecture at least for a common special case. In Section~\ref%
{subsec:primetile} we focus on tiles of prime size. The reason is that these
tiles have several specific properties. Therefore it looks to us promising
to try to prove the Golomb--Welch conjecture for these tiles. We note that
the interested reader can find more details on the results presented in this
section in \cite{H-K}.


\begin{IEEEbiography}{Peter Horak}
  Peter Horak has receied his Ph.D. (1979) and Dr.Sc. (1995) degrees in
  mathematics from the Comenius University in Bratislava, Slovakia. His
  research interest include graph theory, combinatorics, coding theory,
  cryptography, and theoretical computer science. He also has a paper on number
  theory and another one on topology. These two papers have been published in
  American Mathematical Monthly. Peter Horak has held a permanent or a visiting
  positions at several universities in (Czecho)Slovakia, USA, Canada, and
  Kuwait. Since 2003 he is a professor of mathematics at University of
  Washington, Tacoma. He solved (with coauthors) three problems of Paul
  Erd\H{o}s, and a problem of Donald Knuth. His Erd\H{o}s number is $1$. 
\end{IEEEbiography}

\begin{IEEEbiography}{Dongryul Kim}
  Dongryul Kim is an undergraduate student studying at Harvard College. He
  graduated from Seoul Science High School in 2015. He won three gold medals
  from the International Mathematical Olympiad in 2012, 13, and 14. In 2016,
  Dongryul Kim took part in the William Lowell Putnam competition and finished
  among the first five students (unranked). 
\end{IEEEbiography}

\end{document}